\documentclass[journal, onecolumn,]{IEEEtran}

\usepackage{amsmath,amsfonts,amssymb,amsthm}
\usepackage{cite}
\usepackage{algorithmic}
\usepackage{algorithm}
\usepackage[caption=false,font=footnotesize]{subfig}
\usepackage{graphicx}
\usepackage{subfig}
\DeclareGraphicsExtensions{.pdf,.jpeg,.png}
\usepackage{booktabs}
\usepackage{xcolor}
\usepackage{microtype}
\usepackage{array}
\usepackage[switch]{lineno} 
\usepackage{physics}
\usepackage[normalem]{ulem}
\usepackage{palatino}

\newcommand{\textbm}[1]{\textbf{\texttt{\color{white!25!black}{#1}}}}

\newtheorem{theorem}{Theorem}
\newtheorem{proposition}{Proposition}
\newtheorem{definition}{Definition}
\newtheorem{assumption}{Assumption}

\newtheorem{lemma}{Lemma}
\newtheorem{remark}{Remark}

\newcommand{\Sy}{\mathbb{S}}
\renewcommand{\Re}{\mathbb{R}}

\newcommand{\inner}[2]{\langle #1, #2 \rangle}
\renewcommand{\leq}{\leqslant}

\def\0{\mathbf{0}}
\def\1{\mathbf{1}}

\def\beps{\boldsymbol{\epsilon}}
\def\b{\boldsymbol{b}}
\def\e{\boldsymbol{e}}

\def\h{\boldsymbol{h}}
\def\p{\boldsymbol{p}}
\def\q{\boldsymbol{q}}

\def\u{\boldsymbol{u}}
\def\v{\boldsymbol{v}}
\def\w{\boldsymbol{w}}
\def\x{\boldsymbol{x}}
\def\y{\boldsymbol{y}}
\def\z{\boldsymbol{z}}
\def\ze{\mathbf{0}}

\def\A{\mathbf{A}}
\def\B{\mathbf{B}}
\def\E{\mathbf{E}}
\def\F{\mathbf{F}}
\def\G{\mathbf{G}}
\def\H{\mathbf{H}}
\def\I{\mathbf{I}}
\def\J{\mathbf{J}}

\def\K{\mathbf{K}}
\def\M{\mathbf{M}}
\def\N{\mathbf{N}}
\def\P{\mathbf{P}}

\def\R{\mathbf{R}}
\def\S{\mathbf{S}}

\def\V{\mathbf{V}}
\def\D{\mathbf{D}}
\def\W{\mathbf{W}}
\def\sqrtD{\D^{\frac{1}{2}}}
\def\invsqrtD{\D^{-\frac{1}{2}}}

\def\Re{\mathbb{R}}
\def\Z{\mathbb{Z}}

\newcommand{\prox}{\mathrm{prox}}
\newcommand{\fix}{\mathrm{fix}}
\DeclareMathOperator*{\diag}{diag}

\def\cN{\mathcal{N}}

\begin{document}

\title{Linear Convergence of Plug-and-Play Algorithms with Kernel Denoisers}

\author{Arghya~Sinha,
        Bhartendu~Kumar,
        Chirayu~D.~Athalye,
        and~Kunal~N.~Chaudhury
        
\thanks{A.~Sinha and K.~N.~Chaudhury are with the Indian Institute of Science, Bengaluru: 560012, India. B.~Kumar was at IISc during this research work. C.~D.~Athalye is at BITS Pilani, K~K~Birla Goa Campus, Goa: 403726, India. Correspondence: kunal@iisc.ac.in.}
\thanks{A. Sinha was supported by the PMRF fellowship TF/PMRF-22-5534 from the Government of India.

C. D. Athalye was supported by BITS Pilani under the NFSG Grant.

K.~N.~Chaudhury was supported by grant STR/2021/000011 from the Government of India.}
}

\markboth{}%
{Sinha \MakeLowercase{\textit{et al.}}: Convergence for Kernel Denoiser}

\maketitle

\begin{abstract}
The use of denoisers for image reconstruction has shown significant potential, especially for the Plug-and-Play (PnP) framework. In PnP, a powerful denoiser is used as an implicit regularizer in proximal algorithms such as ISTA and ADMM. The focus of this work is on the convergence of PnP iterates for linear inverse problems using kernel denoisers. It was shown in prior work that the update operator in standard PnP is contractive for symmetric kernel denoisers under appropriate conditions on the denoiser and the linear forward operator. Consequently, we could establish global linear convergence of the iterates using the contraction mapping theorem. In this work, we develop a unified framework to establish global linear convergence for symmetric and nonsymmetric kernel denoisers. Additionally, we derive quantitative bounds on the contraction factor (convergence rate) for inpainting, deblurring, and superresolution. We present numerical results to validate our theoretical findings.
\end{abstract}


\IEEEpeerreviewmaketitle

\section{Introduction}
\label{sec:intro}

\IEEEPARstart{T}{he} use of denoisers for iterative image reconstruction was pioneered by the seminal works \cite{romano2017little,sreehari2016plug}, where classical image denoisers such as  NLM \cite{buades2010image} and BM3D \cite{dabov2007image} were used for image reconstruction. Later, it was shown that one could get state-of-the-art reconstructions using trained denoisers~\cite{Dong2018_DNN_prior,zhang2021plug}. Thus, PnP enables the integration of deep image priors with model-based inversion \cite{bouman2022foundations}. However, it is also known that standard pretrained denoisers, such as DnCNN~\cite{zhang2017beyond}, may cause the reconstruction process to diverge~\cite{terris2024equivariant,nair2022plug,reehorst2018regularization}.  Consequently, providing theoretical guarantees for such denoisers is challenging and often requires redesigning or retraining the network~\cite{CWE2017,reehorst2018regularization,Ryu2019_PnP_trained_conv,raj2019gan,gavaskar2020plug,liu2021study,gavaskar2021plug,cohen2021has,gavaskar2023PnPCS,hurault2022gradient,pesquet_learning_2021,cohen2021regularization,proximal-hurault22a,goujon_learning_2024}.

In this work, we revisit the Plug-and-Play (PnP) method \cite{sreehari2016plug}, where regularization is performed by replacing the proximal operator in algorithms such as ISTA and ADMM \cite{parikh2014proximal} with an off-the-shelf denoiser. More specifically, consider the problem of recovering an image  $\x \in \Re^n$ from linear measurements 
\begin{equation}
\label{eq:fm}
\b = \A \x + \boldsymbol{\beps}, 
\end{equation}
where $\A \in \Re^{m \times n}$ is the forward operator, $\b \in \Re^m$ is the observed image, and $\boldsymbol{\beps} \in \Re^m$ is white Gaussian noise. The standard inversion framework involves the loss function 
\begin{equation}
\label{eq:loss}
f(\x) = \frac{1}{2} \lVert \A\x-\b \rVert_2^2
\end{equation}
and a convex regularizer $g \colon \Re^n \to (-\infty,\infty]$. The reconstruction is performed by considering the optimization problem
\begin{equation}
\label{eq:fplusg}
\min_{\x \in \Re^n} \  f(\x) + g(\x),
\end{equation}
and solving it iteratively using proximal algorithms \cite{parikh2014proximal,beck2017book}. For example, starting with $\x_0 \in \Re^n$, the updates in ISTA are given by
\begin{equation}
\label{eq:ista}
 \x_{k+1} = \prox_{\gamma g} \big(\x_k - \gamma \nabla \! f(\x_k) \big),
\end{equation}
where $\gamma$ is the step size and $ \prox_{\gamma g}$ is the proximal operator of $\gamma g$~\cite{beck2017book},
\begin{equation}
\label{eq:proxdefn}
\prox_{\gamma g} (\x) = \underset{\z \in \Re^n}{\arg \min} \ \left\{  \frac{1}{2} \| \z - \x \|_2^2 + \gamma g(\z) \right\}.
\end{equation}

We can alternatively solve \eqref{eq:fplusg} using ADMM, which is known to exhibit better empirical convergence \cite{parikh2014proximal}. Starting with $\y_0, \z_0 \in \Re^n$ and some fixed parameter $\rho > 0$, the updates in ADMM are given by
\begin{align}
\label{eq:admm}
\x_{k+1} &= \prox_{\rho f} (\y_k - \z_k),  \nonumber\\ 
\y_{k+1} &= \prox_{\rho g} (\x_{k+1} + \z_k),  \\ 
\z_{k+1} &= \z_k + \x_{k+1} - \y_{k+1}. \nonumber
\end{align}

Both ISTA and ADMM are known to converge to a minimizer of \eqref{eq:fplusg} under mild technical assumptions~\cite{beck2017book,eckstein1992douglas}.

\subsection{PnP algorithms}

PnP is founded on a simple yet profound insight: rather than explicitly defining the regularizer $g$ and computing its proximal operator in \eqref{eq:ista} and \eqref{eq:admm}, one can directly replace the proximal operator with a powerful denoiser. After all, the proximal operator of $g$ inherently serves as a denoising mechanism in these algorithms. More specifically, we perform the update as
\begin{equation}
\label{eq:pnp-ista}
 \x_{k+1} = \W \big(\x_k- \gamma \nabla \!f(\x_k)\big),
\end{equation}
where $\prox_{\gamma g}$ in \eqref{eq:ista} is replaced by a denoising operator $\W$. This is referred to as \textbm{PnP-ISTA}. Similarly, in the PnP variant of \eqref{eq:admm}, called \textbm{PnP-ADMM}, the updates are performed as follows:
\begin{align}
\label{eq:pnp-admm}
\x_{k+1} &= \prox_{\rho f} (\y_k - \z_k),  \nonumber\\ 
\y_{k+1} &= \W (\x_{k+1} + \z_k),  \\ 
\z_{k+1} &= \z_k + \x_{k+1} - \y_{k+1}. \nonumber
\end{align}

As suggested by the notation in \eqref{eq:pnp-ista} and \eqref{eq:pnp-admm}, our focus is on the use of linear denoisers for $\W$, such as NLM  \cite{buades2010image}, LARK \cite{takeda2007kernel}, DSG-NLM \cite{sreehari2016plug}, and GLIDE \cite{talebi2013global}. This combination of linear inverse problems and linear denoisers is the simplest PnP model that works well in practice \cite{sreehari2016plug,nair2022plug,nair2019hyperspectral} and is analytically tractable \cite{ACK2023-contractivity,AC2023-correction,gavaskar2021plug}.  A detailed discussion of linear denoisers is provided in Section~\ref{sec:kd}.

PnP gives state-of-the-art results with trained denoisers \cite{zhang2017beyond,zhang2021plug,Ryu2019_PnP_trained_conv,zhang2017learning}. 
However, from an operator-theoretic perspective, it is not immediately clear why the sequence generated by \eqref{eq:pnp-ista} should converge, let alone produce good results, unless the trained denoiser is proximable or an averaged operator~\cite{bauschke2017convex}. In particular, it has been shown that trained denoisers can result in divergence of the PnP iterations~\cite{terris2024equivariant,reehorst2018regularization,nair2024averaged} and, thus, a convergence guarantee can act as a safeguard. The convergence aspect of PnP is an active area of research \cite{CWE2017,reehorst2018regularization,Ryu2019_PnP_trained_conv,raj2019gan,gavaskar2021plug,cohen2021has,gavaskar2023PnPCS}.

In this paper, we complete the work started in \cite{ACK2023-contractivity}, where it is shown that \textbm{PnP-ISTA} iterates converge linearly for symmetric denoisers \cite{sreehari2016plug,talebi2013global,milanfar2013symmetrizing}. This is based on the observation that \eqref{eq:pnp-ista} can be viewed as a fixed-point iteration $\x_{k+1} = \P(\x_k)$, where the update operator $\P$ is derived from the forward and denoising operators $\A$ and $\W$. Similarly, the $\{\y_k\}$ iterates in \eqref{eq:pnp-admm} can be written as $\y_k =\W\u_k$, where $\u_k$ is given by the fixed-point iteration $\u_{k+1} = \R(\u_k)$, and $\R$ is again derived from $\A$ and $\W$. The main result in \cite{ACK2023-contractivity} is that if $\W$ is a symmetric denoiser and $\A$ corresponds to inpainting, deblurring, or superresolution, then  $\P$ and $\R$ are contractive for any $0 < \gamma < 2$ and $\rho >0$ (see Section~\ref{sec:SPLresults}). Subsequently, the contraction mapping theorem \cite{bauschke2017convex} was used to guarantee global linear convergence of \eqref{eq:pnp-ista} and \eqref{eq:pnp-admm} to a unique reconstruction. Notably, the contractivity of $\P$ and $\R$ was deduced from niche properties of $\W$ such as positive semidefiniteness and irreducibility. 

\subsection{Contributions} 

The present work is motivated by the question of whether the results in \cite{ACK2023-contractivity} can be extended to a broader range of kernel denoisers, such as NLM \cite{buades2010image}, {which are neither symmetric nor nonexpansive with respect to the standard Euclidean inner product}, but are faster than their symmetric counterparts \cite{sreehari2016plug}. Our main results are as follows.

\begin{enumerate}

\item We present a counterexample to show that the update operator in \textbm{PnP-ISTA} is generally not contractive if $\W$ is nonsymmetric, i.e., the analysis in \cite{ACK2023-contractivity} cannot directly be extended to kernel denoisers, even if we work with a different norm.

\item Subsequently, motivated by prior work \cite{gavaskar2021plug}, we consider scaled variants of \textbm{PnP-ISTA} and \textbm{PnP-ADMM} for (nonsymmetric) kernel denoisers and prove that the respective update operators are contractive. However, the contractivity in question is not w.r.t. the Euclidean norm but w.r.t. a quadratic norm derived from $\W$. This way, we can extend the linear convergence guarantee in \cite{ACK2023-contractivity} from symmetric denoisers to nonsymmetric denoisers. In particular, we develop a unified convergence theory for symmetric and nonsymmetric denoisers.

\item We derive bounds on the contraction factor in terms of problem parameters such as the subsampling rate, the spectral gap of $\W$, and parameters $\gamma$ and $\rho$. It turns out that the actual contraction factors are barely less than $1$, and a delicate analysis is required to compute bounds that are strictly less than $1$. In particular, we study the effect of various measurement and algorithmic parameters on the actual and estimated contraction factors.
\end{enumerate}

\subsection{Organization}
 
In Sec.~\ref{sec:back}, we provide the necessary background on kernel denoisers and discuss our prior work. We motivate the problem and discuss scaled PnP algorithms that go naturally with kernel denoisers in Sec.~\ref{sec:scaledPnP}.
We present our main results in Sections~\ref{sec:ctrscaledpnp} and \ref{sec:boundctr}, which are numerically validated in Sec.~\ref{sec:exp}. 

\subsection{Notation} 
We denote the all-ones vector in $\Re^n$ by $\e$. The spectrum and the spectral norm of $\M$ are denoted by $\sigma(\M)$ and $\|\M\|_2$. We use $\Sy^n$ to denote the space of $(n \times n)$ symmetric matrices and $\Sy^n_+$ to denote the set of $(n \times n)$ symmetric positive semidefinite matrices. We use $\fix(\M)$ to denote the fixed points of $\M$, i.e., $\fix(\M) = \{\x: \, \M\x=\x\}$. For any $\v \in \Re^n$, we use $\diag(\v)$ to represent a diagonal matrix with $\v$ as its diagonal elements; additionally, we use $\diag(\M)$ to denote the vector consisting of the diagonal elements of $\M$. We use $\M \sim \N$ if $\M$ and $\N$ are similar. Unless stated otherwise, we denote an abstract inner product on $\mathbb{R}^n$ as $\langle \cdot, \cdot \rangle$ and the associated norm as $\|\cdot\|$. In particular, we use $\langle \cdot, \cdot \rangle_2$ for the dot product on $\mathbb{R}^n$ and $\|\cdot\|_2$ for the corresponding norm.

\section{Background}
\label{sec:back}

\subsection{Kernel denoiser}
\label{sec:kd}
 
A kernel denoiser is abstractly a linear transform that takes a noisy image and returns the denoised output.   More precisely, let $X \colon \Omega \to [0,1]$ be the input image, where $\Omega  \subset \Z^2$ is a finite square grid, and let $X_{\p}$ be the intensity value at location $\p \in \Omega$. 
A kernel denoiser is specified using a kernel function $\phi \colon \Omega \times \Omega \to \Re$, where $\phi(\p,\q)$ measures the affinity between pixels $\p, \q \in \Omega$ \cite{milanfar2013symmetrizing}. The denoiser output $\W(X)$ is given by
\begin{equation}
\label{neighbor}
\forall \ \p \in \Omega: \quad (\W (X))_{\p} = \frac{\sum_{\q \in \Omega} \phi(\p, \q) \, X_{\q}}{\sum_{\q \in \Omega} \phi(\p, \q)}.
\end{equation}
In other words, the output at a given pixel is obtained using a weighted average of neighboring pixels, where the weights are derived from the affinity between pixels. 

From \eqref{neighbor}, it is clear that \(\W\) is a linear operator on a finite-dimensional vector space. Hence, it is convenient to express \eqref{neighbor} in matrix form.  More precisely, if we linearly index $\Omega$ and let $\x \in \Re^n$ be the vector representation of  $X$ corresponding to this indexing, we can represent \eqref{neighbor} as
\begin{equation}
\label{defkernel}
\W = \D^{-1} \K, \qquad \D = \diag(\K \e),
\end{equation}
where $\K$ is called the kernel matrix. We refer the reader to \cite{milanfar2013symmetrizing} for a detailed account of this class of denoisers. Apart from denoising \cite{buades2010image,takeda2007kernel,nair2019hyperspectral,unni2020_pnp_registration}, kernel denoisers have also been used for graph signal processing \cite{yazaki2019interpolation,nagahama2021graph}.

We note that there are subtle variations across the literature in the definition of $\K$ for NLM \cite{sreehari2016plug,buades2010image,milanfar2013tour}. 
These variations do not significantly impact the denoising performance but can affect the mathematical properties. To avoid confusion, we explicitly state the properties required for our analysis.
\begin{assumption}
\label{assumpK}
The kernel matrix $\K$ has these properties: 
\begin{enumerate}
\item $\K \in \Sy^n_+$.
\item $\K$ is nonnegative and irreducible. 
\item $\K_{ii} = 1$ for $1 \leqslant i \leqslant n$.
\end{enumerate}
\end{assumption}
In particular, since $\K$ is nonnegative and $\K_{ii} = 1$, $\D_{ii} \geqslant 1$. Thus, $\D$ is invertible, which is implicitly used in \eqref{defkernel}. 

A commonly used example of a kernel function in the literature is 
\begin{equation}
\label{phi}
\phi(\p,\q) = G\big(\zeta(\p)-\zeta(\q)\big) \, \mathrm{Hat}(\p-\q),
\end{equation}
where $G$ is an isotropic multivariate Gaussian with standard deviation $h$, $\zeta(\p)$ is a patch around $\p$ extracted from a guide image, and $\mathrm{Hat}$ is a hat function supported on a square window $\Theta$ around the origin ($\Theta$ is used to restrict the averaging to a local window). The guide is the input image or some surrogate of the ground truth estimated from the observed image.

For the kernel $\phi$ in \eqref{phi}, we can verify that the corresponding $\K$ follows Assumption \ref{assumpK}. Clearly, $\K$ is symmetric, nonnegative, and $\K_{ii} = 1$ in this case.
Furthermore, we can show that $\K$ is irreducible if $|\Theta| > 1$. Indeed, this follows from the path-based characterization of irreducibility \cite{meyer2000matrix}: for any two pixels $\p,\q \in \Omega$, we can find a sequence of pixels $\p=\p_1, \p_2, \ldots, \p_N = \q$ such that $\p_{t+1} \in \p_t + \Theta$ for  $t=1,\ldots,N-1$. On the other hand, the property $\K \in \Sy^n_+$ is more subtle; this property is not guaranteed if a box function is used in place of $\mathrm{Hat}$, as done in the standard setting~\cite{buades2010image}. That $\K \in \Sy^n_+$ follows from Bochner's theorem~\cite{katznelson2004introduction}, and we refer the reader to \cite[Appendix B]{sreehari2016plug} for a detailed explanation. In short, both $G(\cdot)$ and $\mathrm{Hat}(\cdot)$ are positive definite functions since they are the Fourier transforms of nonnegative functions. Consequently, the corresponding kernel matrices \(\K_G\) and \(\K_{\mathrm{Hat}}\) are positive semidefinite, and thus their Hadamard product \(\K = \K_G \odot \K_{\mathrm{Hat}}\) is also positive semidefinite.

As a consequence of Assumption~\ref{assumpK}, we have the following structure on $\W$.
\begin{proposition}
\label{propertiesW}
If Assumption~\ref{assumpK} holds, then
\begin{enumerate}
\item[(1)] $\W$ is nonnegative and irreducible. (Perron matrix)
\item[(2)]  $\W\e=\e$.  (stochastic matrix) 
\item[(3)] Eigenvalues of $\W$ are $1=\lambda_1 > \lambda_2 \geqslant \cdots \geqslant \lambda_n \geqslant 0$. 
\end{enumerate}
\end{proposition}

Properties (1) and (2) follow from the properties of $\K$ and the fact that $\D_{ii} > 0$. To deduce (3), we write \eqref{defkernel} as
\begin{equation*}
\W = \invsqrtD \M  \sqrtD, \quad \M:=\invsqrtD \K \invsqrtD.
\end{equation*}
Since $\W \sim \M$, where $\M \in \Sy^n_+$, the eigenvalues of $\W$ must be nonnegative. Furthermore, it follows from (1) and (2) that $\sigma(\W) \subset [0,1]$ and that $\lambda_1=1$ is its largest eigenvalue. 
In fact, since $\W$ is a Perron matrix, it follows from the Perron-Frobenius theorem \cite{meyer2000matrix} that $\lambda_2,\ldots,\lambda_n< 1$. The difference $\lambda_1 - \lambda_2$, which will play an important role in our analysis, is called the ``spectral gap''. This is related to the connectivity of the graph associated with $\W$: the spectral gap is nonzero if and only if the graph is connected, and the gap increases with the connectivity \cite{chung1997spectral}.

Note that while \(\D\) and \(\K\) are symmetric, \(\W\) in~\eqref{defkernel} does not necessarily have to be symmetric. Specifically, while \(\W\) is row-stochastic, it is generally not doubly stochastic. This class of denoisers includes NLM~\cite{buades2010image} and LARK~\cite{takeda2007kernel}, which we call kernel denoisers in this paper. On the other hand, denoisers such as GLIDE~\cite{talebi2013global} and DSG-NLM~\cite{sreehari2016plug} have the additional property that \(\W\) is symmetric, and their definitions differ from \eqref{defkernel}. In particular, the DSG-NLM protocol from \cite[Section IV]{sreehari2016plug} can be applied to any kernel denoiser to obtain a symmetric \(\W\) that satisfies the properties in Proposition~\ref{propertiesW}. We will refer to them as symmetric denoisers to distinguish them from kernel denoisers.

\begin{remark}
A natural question is whether the standard analysis of $\textbf{ISTA}$ can be directly applied to $\textbm{PnP-ISTA}$. Indeed, this approach was taken in the original PnP paper~\cite{sreehari2016plug}. However, this method applies only to symmetric denoisers that are proximable \cite{Moreau1965, chaudhury:hal-05038838}, i.e., the denoiser is the proximal operator of some convex function. Since kernel denoisers are generally non-symmetric, they are not proximable, and therefore, the convergence analysis in~\cite{sreehari2016plug} does not extend to them. Additionally, this has to do with objective convergence, which is generally weaker than iterate convergence––the focus of the present work. In this regard, we note that Krasnosel’skii-Mann (KM) theory~\cite{bauschke2017convex} can establish iterate convergence of \textbf{ISTA}. However, KM theory applies to \textbm{PnP-ISTA} only if the denoiser is proximable or, more generally, an averaged operator. {As we will show later, a kernel denoiser may fail to be nonexpansive and, consequently, cannot be an averaged operator, given that averaged operators are necessarily nonexpansive.} In summary, ensuring objective and iterate convergence for kernel denoisers is challenging. We can establish both forms of convergence for symmetric denoisers, but we cannot generally guarantee linear convergence as done in this work.

We also remark that $\textbm{PnP-ISTA}$ seems to converge with kernel denoisers in practice. However, we have not yet found a formal proof or a counterexample. As evidence of the empirical convergence, we present some reconstructions obtained using the NLM denoiser in \textbm{PnP-ISTA} in Figs.~\ref{fig:deblurring} and \ref{fig:superresolution}.
\end{remark}

\subsection{Prior work}
\label{sec:SPLresults}

It was shown in \cite{ACK2023-contractivity} that \textbm{PnP-ISTA} converges linearly if $\W$ is a symmetric denoiser. This is based on the observation that on substituting \eqref{eq:loss} in  \eqref{eq:pnp-ista}, we get
\begin{equation}
\label{eq:PnPiter}
\x_{k+1} = \P\x_k + \q,
\end{equation}
where 
\begin{equation}
\label{defP}
\P= \W(\I-\gamma \A^\top\!  \A) \quad \mbox{and} \quad   \q =\gamma \W\A^\top \b.
\end{equation}
That is, we can write \eqref{eq:PnPiter} as the fixed-point iteration $\x_{k+1} = T(\x_k)$, where $T(\x) = \P \x+\q$. In particular, $T$ is a contraction if $\| \P\|_2 < 1$. In this regard, we recall the following result \cite{ACK2023-contractivity}.
\begin{theorem}
\label{thm:symWcontr}
If $\W$ is a symmetric denoiser and $\gamma \in (0,2)$, then $\|\P\|_2 < 1$ for inpainting, deblurring, and superresolution. 
\end{theorem}

Consequently, it follows from the contraction mapping theorem~\cite[Thm.~1.48]{bauschke2017convex} that $(\x_k)$ converges linearly to a unique reconstruction for any initialization $\x_0 \in \Re^n$. In other words, \textbm{PnP-ISTA} is guaranteed to converge globally at a geometric rate for inpainting, deblurring, and superresolution. A similar result was also derived for \textbm{PnP-ADMM} in \cite{ACK2023-contractivity}.

Thm.~\ref{thm:symWcontr} is a nontrivial observation since $\W$ and $ \I-\gamma \A^\top\!  \A$ are not contractive w.r.t. the spectral norm. Indeed, $\| \W \|_2=1$ from Prop.~\ref{propertiesW}, and $\| \I-\gamma \A^\top\!  \A \|_2 \geqslant 1$ since $\A$ has a nontrivial nullspace for inpainting, deblurring, and superresolution. Thm.~\ref{thm:symWcontr} can be deduced using a subtle relation between the eigenspaces of these operators.

 We remark that the irreducibility property in Assumption~\ref{assumpK} is crucial for Theorem~\ref{thm:symWcontr}. For example, \(\W = \I\) (though not corresponding to a practical denoiser) satisfies all the properties in Proposition~\ref{propertiesW} except irreducibility.  In this case, $\P=\I-\gamma \A^\top\!  \A$, so that any nonzero vector in the null space of \(\A\) is a fixed point of \(\P\). Consequently, $\P$ cannot be a contraction operator.

A natural question is whether Thm.~\ref{thm:symWcontr} can be extended to a kernel denoiser. It was empirically observed in \cite[Table 1]{ACK2023-contractivity} that $\|\P\|_2 > 1$  for kernel denoisers. This is essentially because of the nonsymmetry of $\W$. Indeed, consider the setting $\A=\I$, corresponding to having full measurements. In this case, $\P=(1-\gamma) \W$ and $\| \P \|_2 = |1-\gamma| \, \|\W\|_2$. Now, consider a kernel denoiser $\W$ that is not doubly stochastic, i.e., while the row sums equal one, the column sums do not. We can show that $\|\W\|_2 > 1$ for such a denoiser. Subsequently, we have $\| \P \|_2 >1$ for $0 < \gamma < 1- \|\W\|^{-1}_2$. That is, no matter how small $\varepsilon$ is, we can construct a kernel denoiser $\W$ and $\A$ for which $\| \P \|_2 >1$ for some $ \gamma$ in $(0,\varepsilon)$. However, this example does not eliminate the possibility that \( \P \) may be contractive in a different norm.

As an exception, it was shown in \cite{ACK2023-contractivity} that linear convergence can be established for inpainting with kernel denoisers. The key idea was to establish contractivity of $\P$ in a different operator norm. Specifically, consider the operator norm 
\begin{equation}
\label{Dnorm}
\|\P\|_{\D} = \max \, \Big\{\|\P\x\|_{\D}: \ \|\x\|_{\D}=1 \Big\},
\end{equation}
induced by the following inner product and norm on $\Re^n$:
\begin{equation}
\label{D-ip-norm}
 \langle \x, \y \rangle_{\D}=\x^\top \! \D \y,\qquad  \|\x\|_{\D} = \sqrt{\langle \x, \x \rangle}_{\D},
\end{equation}
where $\D$ is given by \eqref{defkernel}. We will refer to \eqref{Dnorm} as the $\D$-norm. The following is a restatement of \cite[Thm.~3]{ACK2023-contractivity}.
\begin{theorem}
\label{thm:kernelWcontr}
If $\W$ is a kernel denoiser and $\gamma \in (0,2)$, then $\|\P\|_{\D} < 1$ for inpainting. 
\end{theorem}
Since the contraction mapping theorem allows for arbitrary norms, we can establish the linear convergence of \eqref{eq:pnp-ista} for a kernel denoiser using Theorem~\ref{thm:kernelWcontr} and the $\|\cdot \|_{\D}$-norm.

\subsection{A counterexample}

Thm.~\ref{thm:kernelWcontr} strongly relies on the diagonal nature of $\A$ for inpainting and cannot be extended to deblurring and superresolution. Using a concrete deblurring example, we will prove that $\|\P\|_{\D} <1 $ cannot be guaranteed with a kernel denoiser. Consider the forward operator $\A = (1/n) \, \e\e^\top$ corresponding to a uniform blur. Consider a hypothetical case where the pixel (and patch) values are $x_1 = \cdots = x_{n-1}=0$ and $x_n = 1$; this is modeled on the pixel distribution near an edge. For any $\varepsilon >0$, consider the kernel matrix $\K_{\varepsilon}$ given by
\begin{equation*}
(\K_{\varepsilon})_{ij} = \phi_{\varepsilon}(x_i - x_j), \quad \phi_{\varepsilon}(t)=\exp(- t^2/\varepsilon^2 ).
\end{equation*}
We can verify that $\K_{\varepsilon}$ satisfies Assumption~\ref{assumpK} for all $\varepsilon>0$. The only nontrivial property is $\K \in \Sy^n_+$; this comes from Bochner's theorem~\cite{katznelson2004introduction} as $\phi_{\varepsilon}$ is the Fourier transform of a nonnegative function. 

Let $\D_{\varepsilon}$ and $\W_{\varepsilon}$ be obtained by letting $\K=\K_{\varepsilon}$ in \eqref{defkernel}, so that $\W_{\varepsilon}$ has the properties in Prop.~\ref{propertiesW}. Since working directly with $\D_{\varepsilon}$ and $\K_{\varepsilon}$ is difficult, we consider their limiting forms:
\begin{equation}
\label{D0}
\D_0 := \lim_{\varepsilon \downarrow 0} \, \D_{\varepsilon} 
= \diag\,(n-1,\ldots,n-1,1),  
\end{equation}
and
\begin{equation}
\label{K0}
\K_0 := \lim_{\varepsilon \downarrow 0} \, \K_{\varepsilon} 
= \begin{pmatrix}
    \H   &   \ze \\
    \ze^\top & 1 
  \end{pmatrix},
\end{equation}
where $\H$ is the $(n-1)  \times (n-1)$ all-ones matrix. If we let $\W_0 = \D^{-1}_0 \K_0$, then 
\begin{equation}
\label{eq:limit}
\lim_{\varepsilon \downarrow 0} \, \| \W_{\varepsilon}  (\I - \A^\top \! \A)\|_{\D_{\varepsilon} } 
= \| \W_0 (\I - \A^\top \!\A)\|_{\D_0}. 
\end{equation}
Moreover, we can show from \eqref{D0} and \eqref{K0} that
\begin{equation*}
\| \W_0 (\I - \A^\top \!\A)\|_{\D_0}   = \frac{1}{n}( 2n^2-4n+4)^{1/2}.
\end{equation*}
The quantity on the right is strictly greater than $1$ for all $n \geqslant 3$. Hence, using the continuity of norm and \eqref{eq:limit}, we can conclude that, for all $n \geqslant 3$, $\| \W_{\varepsilon}  (\I - \A^\top \! \A)\|_{\D_{\varepsilon} } >1$ for some $\varepsilon > 0$. While the above counterexample is somewhat artificial, it tells us that the contractivity of $\P$ cannot generally be guaranteed for deburring using the $\D$-norm.

\section{PnP with Kernel Denoisers}
\label{sec:scaledPnP}

We now explain how iterate convergence of PnP with kernel denoisers can be established. To this end, we consider the class of scaled PnP algorithms introduced in~\cite{gavaskar2021plug}. This work demonstrated that rather than symmetrizing the denoiser using the DSG-NLM protocol—which increases the computational cost–the diagonalizability of kernel denoisers can be leveraged to make them self-adjoint with respect to a special inner product. This requires us to use a slightly modified version of the original PnP algorithm called scaled PnP. In particular, this insight allowed us to interpret a kernel denoiser as a proximal operator, which helped establish objective convergence for scaled PnP~\cite{gavaskar2021plug}. In the following sections, we will develop a common framework to analyze the iterates of \textbf{ISTA} generated by a kernel denoiser that satisfies the properties outlined in Prop.~\ref{propertiesW}, regardless of whether it is symmetric or nonsymmetric. We will use this framework to demonstrate the global linear convergence of the iterates and estimate the convergence rates.

To explain the motivation behind scaled PnP, we observe that although \eqref{defkernel} is nonsymmetric, it is self-adjoint w.r.t. the inner product defined in \eqref{D-ip-norm}. That is, for all $\x, \y \in \Re^n$, 
\begin{equation*}
\langle \W\x,\y \rangle_{\D} = \x^\top \! \K \y=\langle \x,\W\y \rangle_{\D}.
\end{equation*}
Thus, while a symmetric denoiser is self-adjoint w.r.t. the standard dot product on $\Re^n$, a kernel denoiser is self-adjoint w.r.t. the inner product \eqref{D-ip-norm}. This suggests that working in the inner-product space $(\Re^n, \langle \cdot, \cdot \rangle_{\D})$ would be more natural for kernel denoisers. This was the original motivation behind the scaled PnP algorithms in \cite{gavaskar2021plug}, which amounts to reformulating  \textbm{PnP-ISTA} and \textbm{PnP-ADMM} in $(\Re^n, \langle \cdot, \cdot \rangle_{\D})$. In this context, we note that ISTA and ADMM can be developed in the general setting of inner-product spaces \cite{beck2017first,bauschke2017convex}. The former requires us to specify the gradient, and the latter requires the proximal operator, both of which depend crucially on the choice of the inner product. 

We can show from \eqref{D-ip-norm}  that the gradient of \eqref{eq:loss} w.r.t. $\langle \cdot, \cdot \rangle_{\D}$ is $\D^{-1} \, \nabla f(\x_k)$, where $\nabla f$ is the standard gradient operator. As a result, the \textbm{PnP-ISTA} update \eqref{eq:pnp-ista} has to be corrected to
\begin{equation}
\label{eq:scaled-pnp-ista}
\x_{k+1} =\W \big(\x_k - \gamma \D^{-1} \, \nabla f(\x_k) \big),
\end{equation}
where $\W$ is a kernel denoiser. This is referred to as \textbm{Scaled PnP-ISTA} (\textbm{Sc-PnP-ISTA}) in   \cite{gavaskar2021plug}. On substituting \eqref{eq:loss}, we can write \eqref{eq:scaled-pnp-ista} as 
\begin{equation}
\label{fpiter1}
\x_{k+1}=\P_s\x_{k}+\q_s,
\end{equation}
where $\q_s \in \Re^n$ is an irrelevant vector, and 
\begin{equation}
\label{defPs}
\P_s= \W\G_s, \quad  \G_s = (\I-\gamma  \D^{-1} \, \A^\top\!  \A).
\end{equation}

On the other hand, by working with the proximal operator in $(\Re^n, \langle \cdot, \cdot \rangle_{\D})$, we obtain the scaled counterpart of \eqref{eq:pnp-admm}. More precisely, we need to work with 
\begin{equation}
\label{eq:Dproxdefn}
\prox_{\D,\rho f} (\x) = \underset{\z \in \Re^n}{\arg \min} \ \left\{  \frac{1}{2} \| \z - \x \|_{\D}^2 + \rho f(\z) \right\},
\end{equation}
which is obtained by replacing $\| \cdot \|_2$ in \eqref{eq:proxdefn} with the $\D$-norm. In particular, \textbm{Sc-PnP-ADMM} updates  are given by
\begin{align}
\label{eq:scaled-pnp-admm}
\x_{k+1} &=\prox_{\D,\rho f}  (\y_k - \z_k),  \nonumber\\ 
\y_{k+1} &= \W (\x_{k+1} + \z_k),  \\ 
\z_{k+1} &= \z_k + \x_{k+1} - \y_{k+1}. \nonumber
\end{align}
On substituting the loss \eqref{eq:loss} in \eqref{eq:Dproxdefn}, we get
\begin{equation}
\label{D-prox}
\prox_{\D,\rho f} (\x) = \left(\I+\rho \D^{-1} \A^\top \A \right)^{-1}\x + \w,
\end{equation}
where $\w \in \Re^n$ does not depend on $\x$. Importantly, we can write the iterates $\{\y_k\}$  in \eqref{eq:scaled-pnp-admm} as
\begin{equation}
\label{fpiter2}
\y_k =\W\u_k, \quad \u_{k+1} = \R_s \u_k + \h_s,
\end{equation}
where $\h_s$ is an irrelevant quantity. We refer the reader to \cite[Prop.~1]{ACK2023-contractivity} for the precise definition of $\{\u_k\}$. This standard reduction comes from the Douglas-Rachford splitting used in ADMM~\cite{parikh2014proximal}. Importantly, using \eqref{D-prox}, we can show that
\begin{equation}
\label{defRs}
\R_s = \frac{1}{2}\left(\I+\J_s \right), \qquad \J := \F_s \V
\end{equation}
where 
\begin{equation}
\label{def:FV}
\F_s = 2\big(\I+\rho\D^{-1}\A^\top\! \A \big)^{-1}-\I \quad \mbox{and} \quad \V = 2\W-\I.
\end{equation}
Similar to \eqref{defPs}, the operator \eqref{defRs} involves the composition of two linear operators, one determined by the forward model and the other by the denoiser. However, the operator in \eqref{defRs} is more complicated. Note that, unlike in the symmetric case, the operators $\G_s$ and $\F_s$ are controlled by the denoiser via $\D$.

\section{Linear Convergence of PnP}
\label{sec:ctrscaledpnp}

As was observed in \cite{gavaskar2021plug}, we can associate \textbm{Sc-PnP-ISTA}  and \textbm{Sc-PnP-ADMM} with a convex optimization problem of the form \eqref{eq:fplusg}. Consequently, we can establish iterate convergence from standard results  \cite{beck2017first,eckstein1992douglas}. Moreover, if the loss function \eqref{eq:loss} is strongly convex, then the iterates can converge linearly \cite{beck2017first}. However, since $\A$ has a nontrivial nullspace for linear inverse problems \cite{bouman2022foundations}, the loss function cannot be strongly convex. At this point, we clarify what linear convergence means. A sequence $\{\x_k\} \subset \Re^n$ is said to converge linearly if there exist $C>0, \, \delta \in [0,1)$, and some norm $\|\cdot\|$ on $\Re^n$ such that 
\begin{equation}
\label{linconv}
 \forall \, k \geqslant 0:  \quad\|\x_{k+1} - \x_k \| \leqslant C \, \delta^k.
\end{equation}
It follows from \eqref{fpiter1} that $\x_{k+1} - \x_k = \P_s(\x_{k} - \x_{k-1})$. Hence, if we can show that $\|\P_s\|_{\D} < 1$, we would have \eqref{linconv} with $\delta = \|\P_s\|_{\D}$. Similarly, since $\|\W\|_{\D} = 1$ for a kernel denoiser, we have from \eqref{fpiter2} that 
\begin{equation*}
\|\y_{k+1} - \y_k\|_{\D}   =  \|\W(\u_{k} - \u_{k-1})\|_{\D} \leqslant    \|\u_{k} - \u_{k-1}\|_{\D}.
\end{equation*}
Now, recall that $\u_{k+1} = \R_s \u_k$. Hence, if we can show that $\|\R_s\|_{\D} < 1$, we would have \eqref{linconv} for $\{\y_k\}$ with $\delta = \|\R_s\|_{\D}$.

For linear inverse problems such as inpainting, deblurring, and superresolution, the forward operator $\A$ is given by a mix of lowpass blur (convolution) and subsampling operations \cite{bouman2022foundations}. For inpainting, $\A \in \Re^{n \times n}$ is a diagonal matrix with $\A_{ii} \in \{0,1\}$. On the other hand, $\A \in \Re^{n \times n}$ is typically a Gaussian or uniform blur in image deblurring. For superresolution, $\A=\S\B$, where $\S \in \Re^{m \times n}$ is a subsampling operator, and $\B \in \Re^{n \times n}$ is a lowpass blur. We will make the following assumptions that usually hold for such applications.

\begin{assumption}
\label{assumpA}
For inpainting (resp.~superresolution), at least one pixel is observed (resp.~sampled). The blur kernel is nonnegative and normalized for deblurring and superresolution: $\B \geqslant \ze$ and $\B \e=\e$.
\end{assumption}

The above properties will play an essential role in establishing contractivity. Moreover, we will use them to bound the norm of the forward operator.

\begin{remark}
\label{rmk:normA}
It follows from Assumption~\ref{assumpA} that $\|\A\|_2 \leqslant 1$ for inpainting, deblurring, and superresolution. In fact, $\|\A\|_2 = \max_{i} \, |\A_{ii}|= 1$ for inpainting, $\|\A\|_2 =\|\B\|_2=1$  deblurring, and $\|\A\|_2 \leqslant \|\S\|_2 \|\B\|_2= 1$ for superresolution.
\end{remark}

Consider a kernel denoiser $\W$ and a forward operator $\A$ that satisfies Assumption~\ref{assumpA}. Using properties
listed in Prop.~\ref{propertiesW}, we can show that $\|\P_s\|_{\D} \leqslant 1$ and $\|\R_s\|_{\D} \leqslant 1$. However,
this does not necessarily mean they are contractive. Indeed, if there exists a nonzero vector $\u \in \cN(\A)$
which is also a fixed point of $\W$, then we would have
\begin{equation*}
\P_s \u = \W \left( \u - \gamma \D^{-1} \A^\top \!\A \u \right) = \W \u = \u,
\end{equation*}
so that $\P_s$ cannot be contractive. Similarly, if $\u$ is a nonzero vector  in $\cN(\A) \cap \fix(\W)$, then we have $\F_s\u =\u$ and $\V \u=\u$. Consequently, $\R_s \u =\u$, and $\R_s$ cannot be contractive. In other words, for $\P_s$ and $\R_s$ to be contractive, the null space of $\A$ and the fixed points of $\W$ must intersect trivially. We record this property as it plays a fundamental role in our analysis.
\begin{definition}[\textbm{RNP}]
\label{RNP}
Operator $\A$ is said to have the Restricted Nullity Property w.r.t. $\W$ if $\cN(\A) \cap \fix(\W) = \{\ze\}$.
\end{definition}

As noted earlier, $\A$ typically has a nontrivial nullspace, i.e., it is possible that  $\A\x= \ze$ for some nonzero $\x$. However, if we restrict $\x$ to $\fix(\W)$, then \textbm{RNP} stipulates that $\x = \ze$.

\begin{proposition}
\label{prop:RNP}
For inpainting, deblurring, and superresolution, $\A$ has \textbm{RNP} w.r.t symmetric and kernel denoisers.
\end{proposition}

\begin{proof}
We have from Prop.~\ref{propertiesW} that $\fix(\W) = \mathrm{span}({\e})$. Thus, we are done if we can show that $\A\e \neq \ze$. Now, for inpainting, since $\A_{ii}=1$ for some $i$, we have $(\A\e)_i \neq 0$. Also, for deblurring, $\A\e=\e$ and,  for superresolution, $\A\e = \S (\B\e) = \S\e$, which is the all-ones vector in $\Re^m$.
\end{proof}

Remarkably, \textbm{RNP} is sufficient for $\P_s$ and $\R_s$ to be contractive. To establish this, we need the following result. Recall that any self-adjoint operator $\M$ on an inner-product space $(\Re^n,\langle \cdot, \cdot \rangle)$ can be diagonalized in an orthonormal basis~\cite{meyer2000matrix}. Namely, there exists  $\lambda_1,\ldots,\lambda_n \in \Re$ and an orthonormal basis $\q_1,\ldots,\q_n \in \Re^n$ such that $\M\q_i=\lambda_i \q_i$ for all $1 \leqslant i \leqslant n$. We also note that the operator norm of $\M$ is defined as $\norm{\M}=\sup  \, \{\norm{\M\x}: \,  \| \x \|=1\}$.

\begin{lemma}
\label{MN-lemma-1}
Let $\M$ and $\N$ be self-adjoint on $(\Re^n,\langle \cdot, \cdot \rangle)$. Suppose that $\sigma(\M), \sigma(\N) \subset (-1,1]$ and $\fix(\M) \cap \fix(\N) =\{\ze\}$. Then $\|\M\N\| < 1$.
\end{lemma}

This is an adaptation of \cite[Lemma~1]{ACK2023-contractivity}. To keep the account self-contained, we provide a short proof in  Appendix~\ref{pf:MN-lemma-1}. The core idea behind Lemma~\ref{MN-lemma-1} is the following observation.

\begin{proposition}
\label{prop:ctr-fix}
Let $\M$ be self-adjoint on $(\Re^n, \langle \cdot, \cdot \rangle)$ and let $\sigma(\M) \subset (-1,1]$. Then $\| \M \x\| = \|\x\|$ if and only if $\x \in \fix(\M)$, where $\|\cdot \|$ is the norm induced by $\langle \cdot, \cdot \rangle$. 
\end{proposition}

We skip the proof which can easily be deduced from the eigendecomposition of $\M$.

\begin{theorem}
\label{thm:pnpsistactr}
Let $\W$ be a kernel denoiser. Then, for any $ \gamma \in (0,2)$, we have $\|\P_s\|_\D < 1$ for inpainting, deblurring, and superresolution and hence \textbm{Sc-PnP-ISTA} converges linearly.  
\end{theorem}

\begin{proof}

We apply Lemma~\ref{MN-lemma-1} to \eqref{defPs} with $\M=\W$ and $\N = \G_s$.
It is clear that $\M$ and $\N$ are self-adjoint on $(\Re^n, \langle \cdot, \cdot \rangle_{\D})$. Moreover, from Prop.~\ref{propertiesW}, $\sigma(\M) =\sigma(\W)\subset [0,1]$. On the other hand, note that 
\begin{equation}
\label{eq:sim1}
\N \sim \D^{\frac{1}{2}} \N  \D^{-\frac{1}{2}} = \I-\gamma  \H,
\end{equation}
where 
\begin{equation}
\label{eq:defH}
\H=\D^{-\frac{1}{2}}  \A^\top\!  \A \D^{-\frac{1}{2}}.
\end{equation}
Clearly, $\H \in \S^n_+$. Moreover,  $\|\D^{-\frac{1}{2}} \|_2 \leqslant 1$ and $\|\A\|_2 \leqslant 1$ (see~Remark~\ref{rmk:normA}). Hence, 
\begin{equation*}
\|\H\|_2 \leqslant \|\A\|_2^2\, \|\D^{-\frac{1}{2}} \|_2^2 \leqslant 1.
\end{equation*}
Since $\H$ is self-adjoint, we have $\sigma(\H) \subset [0,1]$. Consequently, as $\gamma \in (0,2)$, we see from \eqref{eq:sim1} that $\sigma(\N) = \sigma(\G_s)\subset (-1,1]$. 

Finally, $\fix(\M) = \mathrm{span}({\e})$ from Prop.~\ref{propertiesW}. We claim that $\e \notin \fix(\N)$; this is because $\N \e=\e$ would imply $\A\e =\ze$, which would violate Prop.~\ref{prop:RNP}. Thus, $\fix(\M) \cap \fix(\N)=\{\ze\}$, and the assumptions in Lemma~\ref{MN-lemma-1} are verified. Hence, we have $\|\P_s\|_{\D}=\|\M \N\|_{\D} <1$. 
\end{proof}

We now establish the contractivity of the operator in \textbm{Sc-PnP-ADMM}.

\begin{theorem}
\label{thm:pnpadmmctr}
Let $\W$ be an invertible kernel denoiser. Then, for any $\rho >0$, we have $\|\R_s\|_\D < 1$ for inpainting, deblurring, and superresolution and \textbm{Sc-PnP-ADMM} converges linearly.  
\end{theorem}

\begin{proof}
Note that it suffices to show that $\|\J_s\|_\D < 1$. Indeed, in this case, we have from \eqref{defRs} and the triangle 
inequality that
\begin{equation*}
\forall \, \x \in \Re^n\backslash \{\0\}: \ \|\R_s \x\|_{\D} \leqslant \frac{1}{2} \Big(\|\x\|_{\D} + \|\J_s \x\|_\D \Big) < \|\x\|_{\D}.
\end{equation*}

We apply Lemma~\ref{MN-lemma-1} to $\J_s$ with $\M=\F_s$ and $\N = \V$. It is not difficult to show that $\F_s$ and $\V$ are  
self-adjoint on $(\Re^n, \langle \cdot, \cdot \rangle_{\D})$. Moreover, $\sigma(\N) =\sigma(\V) \subset (-1,1]$. On the other hand, we claim that $\sigma(\F_s) \subset  [ (1-\rho)(1+\rho)^{-1}, \, 1 ]$, which is contained in $(-1,1]$ for all $\rho >0$. Indeed, it follows from \eqref{def:FV} that
\begin{equation}
\label{eq:sim2}
\F_s \sim \D^{\frac{1}{2}} \F_s\D^{-\frac{1}{2}} =  2(\I + \H)^{-1} - \I, 
\end{equation}
where $\H$ is given by \eqref{eq:defH}. However, we have already seen that $\sigma(\H) \subset [0,1]$, and hence our claim follows from \eqref{eq:sim2}.

Finally, note that $\fix(\V) =\fix(\W) = \mathrm{span}({\e})$. Moreover, $\e \notin \fix(\F_s)$, since $\F_s \e=\e$ would give $\A\e =\ze$, violating Prop.~\ref{prop:RNP}. Thus, we have $\fix(\M) \cap \fix(\N)=\{\ze\}$, and it follows from  Lemma~\ref{MN-lemma-1} that $\|\J_s\|_\D < 1$. 
\end{proof}

\begin{remark}
\label{rmk:invertibleW}
It can be shown that $\K$, and consequently $\W$, is nonsingular for NLM and DSG-NLM \cite[Thm.~2.16]{pravinThesis}. Thus, $0 \notin \sigma(\W)$ and $\sigma(\V) \subset (-1,1]$, where the latter is used in the proof of Thm.~\ref{thm:pnpadmmctr}. However, we can show that Thm.~\ref{thm:pnpadmmctr} is valid even when $\W$ is singular. We choose to work with this assumption since the kernel denoisers we employ are invertible, and the analysis is simplified by assuming invertibility.
\end{remark}
\begin{remark}
    {We note that scaled PnP does not change the loss function in~\eqref{eq:loss}, which arises from the linear measurement model and the Gaussian noise assumption. The only modification is in the regularization component, which is symmetrized by introducing the $\D$-inner product as discussed above.}
\end{remark}

\section{Contraction Factor}
\label{sec:boundctr}

We have shown that the update operator in \textbm{PnP-ISTA} and \textbm{PnP-ADMM} is contractive. However, this is only a qualitative result, and ideally, we would like to derive a formula for the contraction factor. Unfortunately, deriving an exact formula is difficult, so we calculate an upper bound instead. This will allow us to understand how the convergence rate is influenced by factors such as the sampling rate, the spectral gap of $\W$, and the algorithmic parameters $\gamma$ and $\rho$. As we will see later, the contraction factor (computed numerically) is close to $1$. Consequently, obtaining a bound that is $<1$ becomes technically challenging. First, we present the following quantitative form of Lemma~\ref{MN-lemma-1} (see Appendix~\ref{pf:MN-lemma-2} for the proof).

\begin{lemma}
\label{MN-lemma-2}
Let $\M$ and $\N$ be self-adjoint on $(\Re^n,\langle \cdot, \cdot \rangle)$ and $\sigma(\M), \sigma(\N) \subset (-1,1]$.  Let the eigenvalues of $\M$ be
\begin{equation}
\label{order-lambda}
|\lambda_1| \geqslant |\lambda_2| \geqslant \cdots \geqslant |\lambda_n|,
\end{equation}
and let $\q_1$ be a unit eigenvector corresponding to $\lambda_1$. Then
\begin{equation}
\label{eq:MNbound}
\|\M \N\|^2 \leqslant (\lambda_1^2-\lambda_2^2) \,\|\N\q_1\|^2 + \lambda_2^2.
\end{equation}
\end{lemma}

 We can use Lemma~\ref{MN-lemma-2} to bound the contraction factor. Given \eqref{eq:MNbound}, we just need to bound the quantity $\|\N\q_1\|^2$, where $\N$ is $\G_s$ for \textbm{PnP-ISTA} and $\F_s$ \textbm{PnP-ADMM}. In particular, the \textbm{RNP} property of $\A$ (Prop.~\ref{prop:RNP}) ensures that $\|\G_s \e\| < 1$ and $\|\F_s \e\| < 1$ (Prop.~\ref{prop:ctr-fix}).

\begin{remark}
\label{rmk:MNswitch}
The bound in \eqref{eq:MNbound} does not change if we switch $\M$ and $\N$. This is because $\norm{\N\M}=\norm{(\N\M)^*}=\norm{\M^*\N^*}=\norm{\M\N}$, where $\A^*$ denotes the adjoint of $\A$. 
\end{remark}

\subsection{Contraction factor for \textbm{ISTA}}

We get similar bounds for \textbm{PnP-ISTA} and \textbm{Sc-PnP-ISTA}. We present the analysis for the latter, as it is more intricate due to the presence of the $\D^{-1}$ factor in $\G_s$.

\begin{proposition}
\label{prop:normGe}
Let $\gamma \in (0,2)$ and $\q_1=\e/\|\e\|_{\D}$. Then
\begin{equation}
\label{eq:boundGse}
\|\G_s \q_1\|^2_\D \leqslant 1- \frac{1}{ \| \D \|_2} \gamma ( 2 - \gamma) \left(\frac{1}{n} \| \A\e \|_2^2\right).
\end{equation}
\end{proposition}

\begin{proof} From the definition of $\G_s$, 
\begin{align}
\label{eq:exp}
\|\G_s \e\|_\D^2 = \| \e-\gamma\D^{-1}\A^\top\!\A\e\|_\D^2.
  \end{align} 
Expanding the RHS, we get
\begin{align*}
\norm{\e}_\D^2 + \gamma^2\norm{\D^{-1}\A^\top\!\A\e}^2_\D  - 2\gamma\inner{\e}{\D^{-1}\A^\top\!\A\e}_\D.
  \end{align*} 
 Now,  $  \inner{\e}{\D^{-1}\A^\top\!\A\e}_\D = \norm{\A\e}_2^2$. As for the third term, since $\norm{\A^\top}_2\leq 1$ (Remark~\ref{rmk:normA}) and $\|\z\|_{\D}=\|\D^{1/2} \z\|_2$, we have
\begin{equation*}
\| \D^{-1}\A^\top\!\A\e\|_\D = \|\D^{-\frac{1}{2}}\A^\top\!( \A\e)\|_2 \leq \|\A\e\|_2.
\end{equation*}
Thus, we can bound \eqref{eq:exp} by $\norm{\e}_\D^2 - \gamma(2-\gamma) \|\A\e\|^2_2$. Finally, we arrive at \eqref{eq:boundGse} by noting that $\| \e\|^2_\D  \leq n\norm{\D}_2$.
\end{proof}

\begin{remark}
\label{rmk:normAe}
For the inverse problems at hand, $\| \A\e \|^2_2$ in \eqref{eq:boundGse} has a precise meaning. For inpainting, $ \| \A\e \|_2^2 =  \mu n$, where $\mu$ is the fraction of observed pixels. On the other hand, $ \| \A\e \|_2^2 = \|\e\|_2^2=  n$ for deblurring, and $ \| \A\e \|_2^2 = \|\S\e\|_2^2 =\mu n$ for superresolution, where $\mu = m/n$ is the subsampling rate. 
\end{remark}

We now bound the contraction factor of the update operator in \textbm{Sc-PnP-ISTA}.

\begin{theorem}
\label{thm:boundISTA}
Let $\W$ be a kernel denoiser and $\gamma \in (0,2)$. Then, for inpainting and superresolution, 
\begin{equation*}
\|\P_s\|^2_{\D} \leqslant  \lambda_2^2 + (1-\lambda_2^2) \Big(1- \frac{1}{ \| \D \|_2} \gamma (2 - \gamma)\mu \Big),
\end{equation*}
and, for deblurring,  
\begin{equation*}
\|\P_s\|^2_{\D} \leqslant  \lambda_2^2 + (1-\lambda_2^2) \Big(1- \frac{1}{ \| \D \|_2} \gamma (2 - \gamma) \Big).
\end{equation*}
\end{theorem}

\begin{proof}
It is evident that $\M=\W$ and $\N=\G_s$ satisfy the assumptions in Lemma~\ref{MN-lemma-2}. In particular, $\lambda_1=1$ and $\q_1 = \e/\|\e\|_{\D}$ for denoiser $\W$. Thus, \eqref{eq:MNbound} becomes
\begin{equation*}
\|\P_s\|_\D^2 \leqslant (1-\lambda_2^2) \,\|\G_s \q_1\|^2 + \lambda_2^2.
\end{equation*}
The desired bounds follow from  Prop.~\ref{prop:normGe} and Remark~\ref{rmk:normAe}.
\end{proof}

\begin{remark}
The bounds in Lemma~\ref{MN-lemma-2} might not always be $<1$. However, the design of the kernel denoisers guarantees that the spectral gap $1 - \lambda_2$ remains positive, ensuring that the bound in Thm.~\ref{thm:boundISTA} is $<1$. For \textbm{ADMM}, we require an additional assumption of invertibility for the denoiser $\W$; this is because it is possible that $\zeta_* 
=-1$ if $\W$ is singular. However, as discussed in Thm.~\ref{thm:pnpadmmctr}, this assumption is
justified, and we obtain a bound strictly less than one for both \textbm{PnP-ADMM} and \textbm{Sc-PnP-ADMM}.
\end{remark}

\begin{remark}
We can get the bound for \textbm{PnP-ISTA} by setting $\D=\I$ in Thm.~\ref{thm:boundISTA}. This gives  
\begin{equation}
\label{bnd1}
\|\P\|^2_2 \leqslant \lambda_2^2 + (1-\lambda_2^2) (1- \gamma)^2
\end{equation}
for deblurring, and
\begin{equation}
\label{bnd2}
\|\P\|^2_2 \leqslant \lambda_2^2 + (1-\lambda_2^2) (1- \gamma (2 - \gamma)\mu)
\end{equation}
for inpainting and superresolution.
\end{remark}

We note that \eqref{bnd1} and \eqref{bnd2} are smaller than the respective bounds for a kernel denoiser.
We can obtain \eqref{bnd1} and \eqref{bnd2} directly from Lemma~\ref{MN-lemma-2} by substituting $\M$ with a symmetric denoiser. Indeed, the operator in this case is $\P$ defined in \eqref{defP}. Setting $\N=\G$ in \eqref{eq:MNbound}, we need to bound $\|\G\q_1\|$, where $\q_1= \e/\sqrt{n}$ is the top eigenvector of $\W$. For deblurring, since $\A\e=\A^\top\! \e=\e$, we have $\G\e=(1-\gamma)\e$, leading to \eqref{bnd1}. Similarly, $\G$ is diagonal for inpainting with elements in $\{1,1-\gamma\}$, and a direct calculation gives \eqref{bnd2}.

\subsection{Contraction factor for  \textbm{ADMM}}

Following \eqref{defRs} and \eqref{def:FV}, we will use $\J, \F, \V$ and $\R$ for \textbm{PnP-ADMM}. As observed earlier in Thm.~\ref{thm:pnpadmmctr}, contractivity of $\R_s$ (resp.~$\R$) follows from the contractivity $\J_s$ (resp.~$\J$). In particular, 
\begin{equation*}
\|\R\|_2 \leqslant \frac{1}{2} \big( 1+\|\J\|_2 \big) \ \mbox{ and } \ \|\R_s\|_\D \leqslant \frac{1}{2} \big( 1+\|\J_s\|_\D \big).
\end{equation*}

Thus, we will focus on bounding $\|\J\|_2$ and $\|\J_s\|_\D$. Bounding $\|\J\|_2$ is particularly easy for inpainting ($\A^\top\! \A$ is diagonal) and deblurring ($\A^\top\! \A\e=\e$). A direct application of Lemma~\ref{MN-lemma-2} on $\J$ requires $\F\e$, which has the following closed-form expression:
\begin{equation}
\label{eq:Fe}
\F\e =
 \begin{cases}
     \e - 2\rho/(1+\rho) \, \diag(\A) & \quad  \text{(inpainting)}, \\
     (1-\rho)/(1+\rho) \e & \quad  \text{(deblurring)},
\end{cases}
\end{equation}
This makes it easy to compute $\|\F\e\|_2$ in Lemma~\ref{MN-lemma-2}.

Recall that the bound for \textbm{PnP-ISTA} involves the spectral gap of $\W$. However, we have $\V$ instead of $\W$ in \textbm{PnP-ADMM}, and $\V$ can have negative eigenvalues. The counterpart of the spectral gap that comes up naturally in this case is
\begin{equation}
\label{def:zeta}
 \zeta_* =\max_{\xi \in \sigma(\V), \, \xi \neq 1} \ |\xi| = \max_{\lambda \in \sigma(\W), \, \lambda \neq 1} \ |2\lambda-1|.
\end{equation}

\begin{figure*}[t]
    \centering
    \subfloat[original]{
    \includegraphics[width=0.16\linewidth]{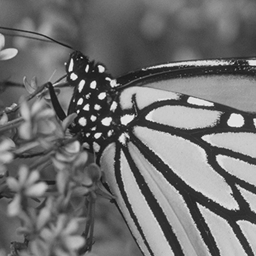}}\hfill
    \subfloat[input]{
    \includegraphics[width=0.16\linewidth]{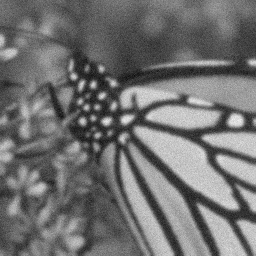}}\hfill
    \subfloat[\textbm{PnP-ISTA}]{
    \includegraphics[width=0.16\linewidth]{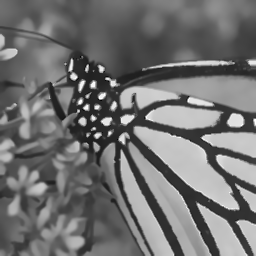}}\hfill
     \subfloat[\textbm{PnP-ISTA} with NLM]{
    \includegraphics[width=0.16\linewidth]{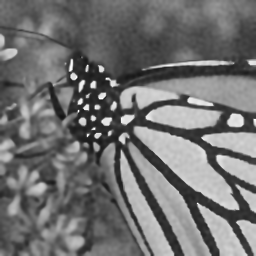}}\hfill
      \subfloat[\textbm{Sc-PnP-ISTA}]{
    \includegraphics[width=0.16\linewidth]{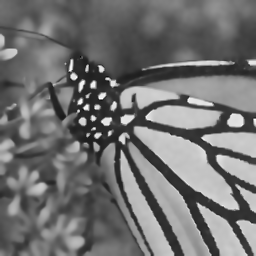}}\hfill
    \caption{Deblurring results using \textbm{PnP-ISTA} with a symmetric DSG-NLM denoiser and a nonsymmetric NLM denoiser, and \textbm{Sc-PnP-ISTA} with a nonsymmetric NLM denoiser. The input image is blurred with a Gaussian kernel of size $25\times 25$ and a standard deviation of $1.6$, followed by the addition of $3\%$ white Gaussian noise. The PSNR values are (b) $22.85$ dB, (c) $29.39$ dB, (d) $28.93$ dB and (e) $28.81$ dB.}
    \label{fig:deblurring}
\end{figure*}

\begin{figure*}[t]
    \centering
    \subfloat[original]{
    \includegraphics[width=0.16\linewidth]{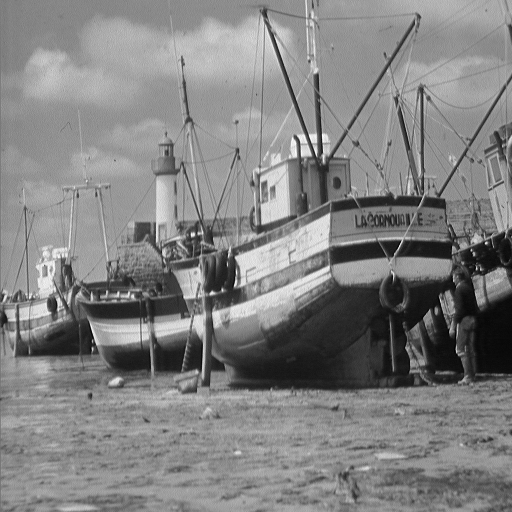}}\hfill
    \subfloat[input]{
    \includegraphics[width=0.16\linewidth]{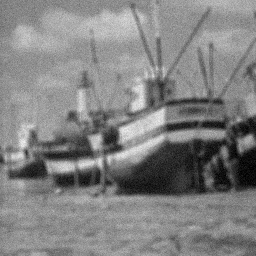}}\hfill
    \subfloat[\textbm{PnP-ISTA}]{
    \includegraphics[width=0.16\linewidth]{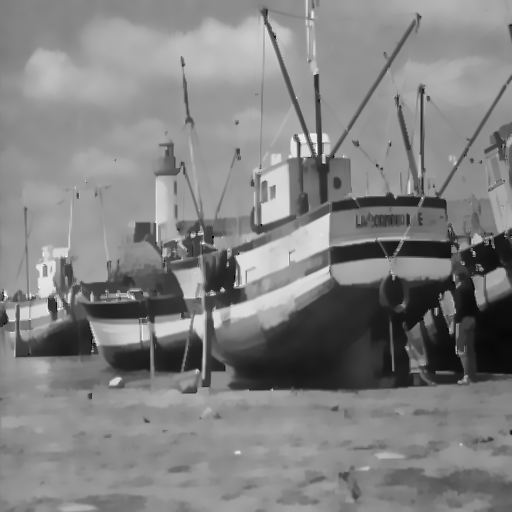}}\hfill
     \subfloat[\textbm{PnP-ISTA} with {NLM}]{
    \includegraphics[width=0.16\linewidth]{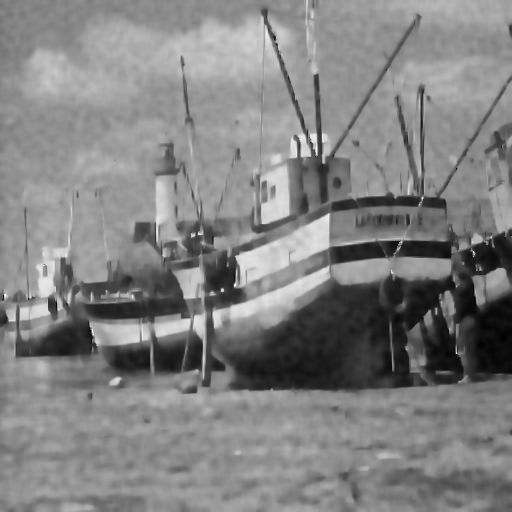}}\hfill
     \subfloat[\textbm{Sc-PnP-ISTA}]{
    \includegraphics[width=0.16\linewidth]{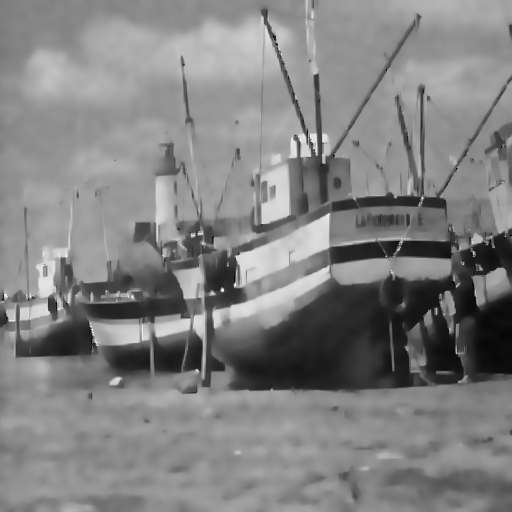}}\hfill
    \caption{$2\times$ superresolution using \textbm{PnP-ISTA} and \textbm{Sc-PnP-ISTA}, with DSG-NLM and NLM denoisers respectively. The input is generated by applying a uniform blur with a $9 \times 9$ kernel, followed by $2\times$ downsampling and adding $3\%$ white Gaussian noise. The PSNR values are (c) $27.31$ dB, (d) $26.67$ dB, and (d) $26.49$ dB.  For reference, the PSNR from bicubic interpolation is $22.80$ dB.}
    \label{fig:superresolution}
\end{figure*}

\begin{theorem}
\label{thm:inp-deblur-sym}
Let $\W$ be a symmetric denoiser and $\rho >0$. Then
\begin{equation}
\label{eq:deblur-bound-J}
    \norm{\J}_2^2 \leq \zeta_*^2 + (1 - \zeta_*^2)\bigg(\frac{1-\rho}{1+\rho}\bigg)^2
\end{equation}
for deblurring, and 
\begin{equation}
\label{eq:inp-bound-J}
    \norm{\J}_2^2 \leq \zeta_*^2 + (1 - \zeta_*^2)\bigg(1-\frac{4\mu\rho}{(1+\rho)^2}\bigg)
\end{equation}
for inpainting, where $\mu$ is the fraction of observed pixels. 
\end{theorem}

\begin{proof}
We will use Lemma~\ref{MN-lemma-2} with $\M=\V$ and $\N=\F=2(\I+\rho \A^\top \!\A)^{-1}-\I$. These are self-adjoint on $(\Re^n,\langle \cdot, \cdot \rangle_2)$. In particular, $|\lambda_1|=1, |\lambda_2|= \zeta_*$, and  $\q_1=\e/\|\e\|_2$. For deblurring, we have from \eqref{eq:Fe} that
\begin{equation*}
\|\F\e\|_2^2= \left( \frac{1-\rho}{1+\rho} \right)^2 \|\e\|_2^2.
\end{equation*}
Substituting $\|\F\q_1\|_2^2$ in Lemma~\ref{MN-lemma-2} and noting that $\J=\M\N$ (also see Remark~\ref{rmk:MNswitch}), we get \eqref{eq:deblur-bound-J}.

For inpainting, let $\Omega=\{i: \A_{ii}=1\}$ be the observed pixels, where $|\Omega|=\mu n$. We have from \eqref{eq:Fe} that
\begin{align*}
\|\F\e\|_2^2  &= \sum_{i \in \Omega} \, (1-\rho)^2/(1+\rho)^2 +  \sum_{i \notin \Omega}\, 1 \\
& = n -  |\Omega| \left( 1 - (1-\rho)^2/(1+\rho)^2 \right) \\
& = \left(1 - 4 \rho \mu/(1+\rho)^2 \right) \|\e\|_2^2.
\end{align*}
Substituting $\|\F\q_1\|_2^2$ in Lemma~\ref{MN-lemma-2}, we get \eqref{eq:inp-bound-J}.
\end{proof}

Unfortunately, we do not have a simple expression for $\F\e$ similar to \eqref{eq:Fe} for superresolution. The situation becomes even more complicated for \textbm{Sc-PnP-ADMM}, where $\F_s$ in  \eqref{def:FV} has an additional $\D^{-1}$ term. Specifically, there is no straightforward expression for $\F_s\e$ in deblurring and superresolution. However, $\F_s$ is diagonal for inpainting, allowing us to easily bound $|\F_s\e|_\D$.

\begin{theorem}
\label{thm:inp-bound-Js}
Let $\W$ be a kernel denoiser. Then, for inpainting with $\mu$ fraction of observed pixels, we have
\begin{equation}
\label{eq:inp-bound-Js}
\norm{\J_s}_\D^2 \leq \zeta_*^2 + (1-\zeta_*^2)\left( 1 -  (1-\theta^2) \frac{\mu}{\|\D\|_2} \right),
\end{equation}
where $\theta := \max_{i}\, (1-\rho \D_{ii}^{-1})/(1+\rho \D_{ii}^{-1}) <1$.
\end{theorem}

The derivation of \eqref{eq:inp-bound-Js} is similar to \eqref{eq:inp-bound-J} and is outlined in Appendix~\ref{pf:inp-bound-Js}.

We will use the following observation to bound $\F_s\e$ for deblurring and superresolution. 

\begin{figure*}[t]
    \centering
    \subfloat[Gaussian deblurring]{
    \includegraphics[width=0.42\linewidth]{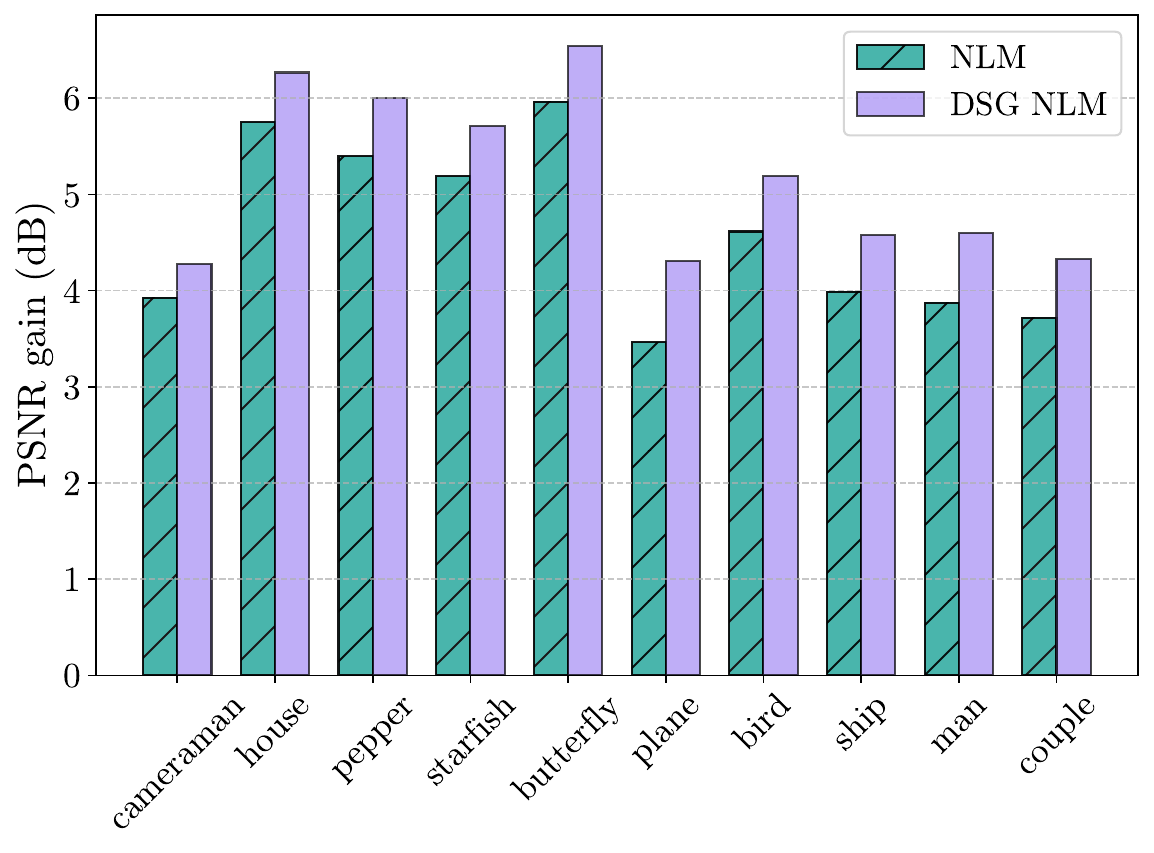}}
    \hfill
    \subfloat[Superresolution]{
    \includegraphics[width=0.42\linewidth]{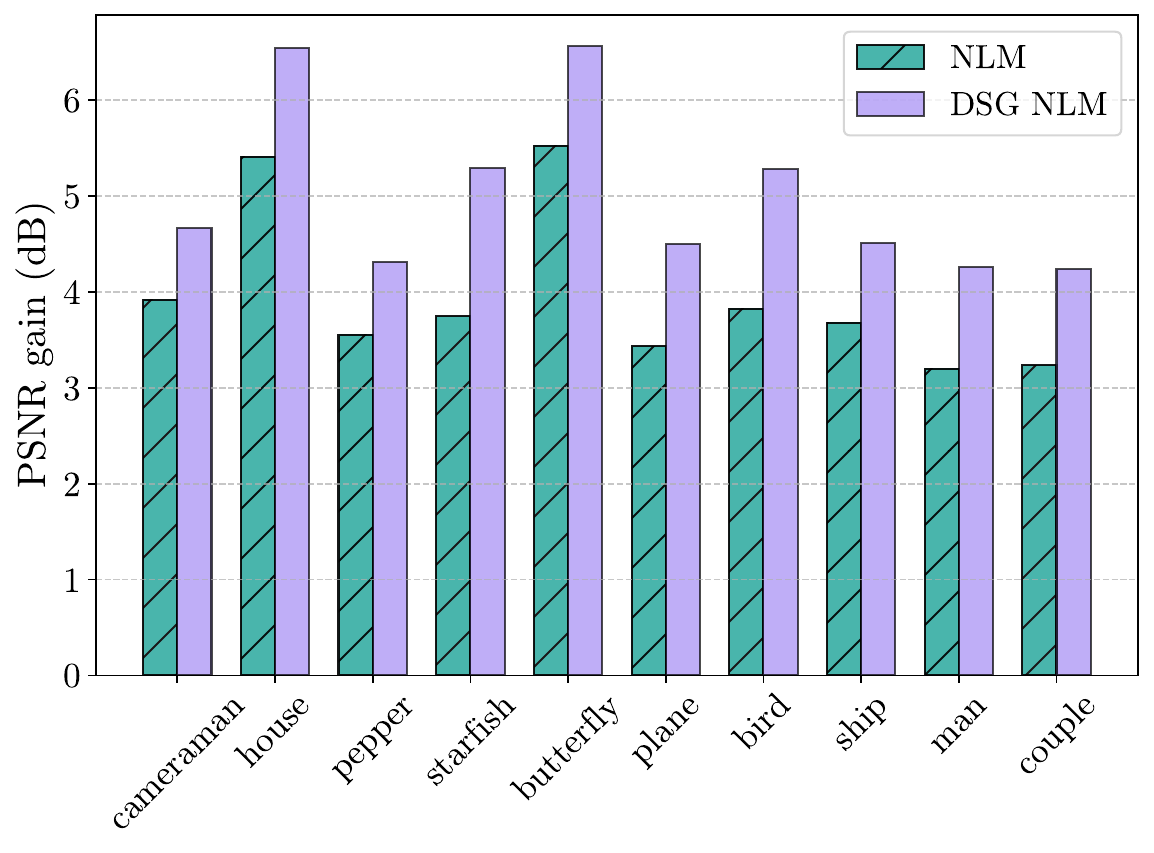}
    }
    \caption{We compare the reconstruction quality obtained using \textbm{ISTA} with symmetric DSG-NLM and nonsymmetric NLM denoisers. The results are averaged over images from the Set12 dataset. For superresolution, the PSNR gain is computed with respect to the bicubic interpolation of the observed image. Overall, the symmetric denoiser DSG-NLM outperforms the nonsymmetric NLM denoiser, although it requires more computations~\cite{sreehari2016plug}.}
    \label{fig:bar_plot_symm_vs_kernel}
\end{figure*}

\begin{proposition}
\label{prop:normHx}
Let $\H$ be a self-adjoint operator on $(\Re^n, \langle \cdot, \cdot \rangle)$ with $\sigma(\H) \subset [-1,1]$. Let $\xi$ be an eigenvalue of $\H$ with unit eigenvector $\w$. Then
\begin{equation*}
\forall \, \x \in \Re^n: \quad \|\H\x\|^2 \leqslant \|\x\|^2 - (1-\xi^2) \langle \w,  \x \rangle^2. 
\end{equation*}
\end{proposition}

We skip the proof of Prop.~\ref{prop:normHx}. This follows directly from the eigendecomposition of $\H$. 

\begin{theorem}
\label{thm:inp-sr-Js}
Let $\W$ be a kernel denoiser. Then 
\begin{equation*}
\norm{\J_s}_\D^2 \leq \zeta_*^2 + (1-\zeta_*^2)\left( 1 -   \frac{1}{n\|\D\|^2_2}\cdot \frac{4 \rho}{(1+\rho)^2} \right)
\end{equation*}
for deblurring, and 
\begin{equation*}
\norm{\J_s}_\D^2 \leq \zeta_*^2 + (1-\zeta_*^2)\left( 1 -   \frac{\mu}{n\|\D\|^2_2}\cdot \frac{4 \rho}{(1+\rho)^2} \right)
\end{equation*}
for superresolution, where $\mu$ is the subsampling rate.
\end{theorem}

\begin{proof}
We can write 
\begin{equation}
\label{def:Fs-B}
\F_s =  2(\I+ \rho \E)^{-1} - \I,  \quad \E:=\D^{-1} \A^\top \! \A.
\end{equation}
Note that $\E$ is self-adjoint on 
$(\Re^n, \langle \cdot, \cdot \rangle_\D)$ and its eigenvalues are nonnegative. In particular, 
$\alpha:=\|\E\|_\D$ is an eigenvalue of $\E$, i.e., there exists $\| \w \|_\D=1$ such that $\E \w = \alpha\, \w$. We also need the fact that $\alpha \leqslant 1$. This is because
\begin{align*}
\|\E\|_\D &= \| \D^{\frac{1}{2}} \E \D^{-\frac{1}{2}} \|_2 \leqslant  \| \D^{-\frac{1}{2}} \|_2^2 \ |\A \|^2_2 \leqslant 1.
\end{align*}
From \eqref{def:Fs-B}, we have $\F_s \w = \xi \, \w$, where $\xi = (1-\rho \alpha)/(1+\rho \alpha)$. Applying Prop.~\ref{prop:normHx} with $\H=\F_s$ and $\x=\e$, we get
\begin{equation}
\label{eq:temp1}
\|\F_s \e \|^2_\D \leqslant \|\e\|_\D^2 - \frac{4\rho \alpha}{(1+\rho \alpha)^2} \langle \w,  \e \rangle_\D^2.
\end{equation}
Recall that we need to bound $\|\F_s \q_1 \|^2_\D$, where $\q_1=\e/\|\e\|_\D^2$. Since $\|\e\|_\D^2 \leqslant n \|\D\|_2$ and $\alpha \leqslant 1$, it follows from \eqref{eq:temp1} that
\begin{equation*}
\|\F_s \q_1 \|^2_\D \leqslant 1 -  \frac{1}{n\|\D\|_2}\cdot \frac{4\rho \alpha}{ (1+\rho)^2} \langle \w,  \e \rangle_\D^2.
\end{equation*}
Since $\A$ (and hence $\E$) has nonnegative components (Assumption~\ref{assumpA}), we can show $w_i \geqslant 0$, where $\w=(w_1,\ldots,w_n)$; see  \cite[Sec.~8.3]{meyer2000matrix} for example. Moreover, as $\D_{ii} \geqslant 1$ and $\| \w \|^2_\D=1$, we have $0 \leqslant w_i \leqslant 1$. Hence,
\begin{equation*}
\langle \w,  \e \rangle_\D =  \sum_{i=1}^n \D_{ii} w_i \geqslant \sum_{i=1}^n \D_{ii} w_i^2 =1.
\end{equation*}
On the other hand,
\begin{align*}
\alpha &= \|\D^{-1} \A^\top \! \A\|_\D   \geqslant \frac{\langle \e, \D^{-1} \A^\top \! \A \e \rangle_\D}{\|\e \|^2_\D} =  \frac{\|\A\e\|_2^2}{\|\e \|^2_\D}
\end{align*}
Consequently, we have from Remark~\ref{rmk:normAe} that
\begin{equation*}
\| \F\q_1\|_\D^2 \leqslant 
 \begin{cases}
   \displaystyle   1 -  \frac{1}{n\|\D\|^2_2}\, \frac{4\rho }{ (1+\rho)^2} & \quad  \text{(deblurring)}, \\
    \displaystyle     1 -  \frac{1}{n\|\D\|^2_2}\, \frac{4\rho \mu}{ (1+\rho)^2}& \quad  \text{(superresolution)}.
\end{cases}
\end{equation*}
Combining this with Lemma~\ref{MN-lemma-2}, we get the desired bounds.
\end{proof}

We note that the bounds in Thm.~\ref{thm:inp-deblur-sym} are not recovered if we set $\D=\I$ in Thm.~\ref{thm:inp-sr-Js}. The extra $1/n$ factor is an artifact of our analysis and might be avoidable.

\subsection{Discussion}

We remark that the bounds in Lemma~\ref{MN-lemma-2} need not be strictly less than one. However, as the spectral gap $1-\lambda_2$ is positive for a kernel denoiser, the bound in Thm. \ref{thm:boundISTA} is guaranteed to be $<1$. For \textbm{ADMM}, we assumed that the denoiser is invertible, which forces the spectral gap $\zeta_*$ in \eqref{def:zeta} to be $<1$. However, as discussed in Remark \ref{rmk:invertibleW}, NLM and DSG-NLM denoisers can be proven to be invertible, and the bound is $<1$ for both \textbm{PnP-ADMM} and \textbm{Sc-PnP-ADMM}.

Our analysis reveals that the contraction factor decreases for both inpainting and superresolution as \(\mu\) increases. This leads to the intuitive conclusion that convergence is faster with more measurements. It is also worth noting
that, for \textbm{PnP-ISTA}, the bound depends on $\W$ through its second-largest eigenvalue
$\lambda_2$. If we view the kernel matrix $\K$ as the adjacency matrix of a graph with the pixels as nodes, then $\lambda_2$ (the spectral gap) represents the connectivity of the graph, indicating how quickly pixels are mixed (or diffused) through repeated applications of $\K$ \cite{milanfar2013tour}. Better connectivity results in a smaller $\lambda_2$ and faster convergence. As shown in Figs.~\ref{fig:thm5-lambda2-effect} and \ref{fig:zeta-effect} in the next section, reducing \(\lambda_2\) and \(|\zeta_*|\) by increasing the bandwidth \(h\) can accelerate convergence. However, while a higher bandwidth improves mixing and speeds up convergence, it can also introduce excessive blurring in the reconstruction (Fig. \ref{fig:convergence-plot}). Thus, there is an optimal tradeoff between convergence speed and reconstruction accuracy.

\begin{figure*}[h]
    \centering
    \subfloat[\textbm{PnP-ISTA}]{
    \includegraphics[width=0.42\linewidth]{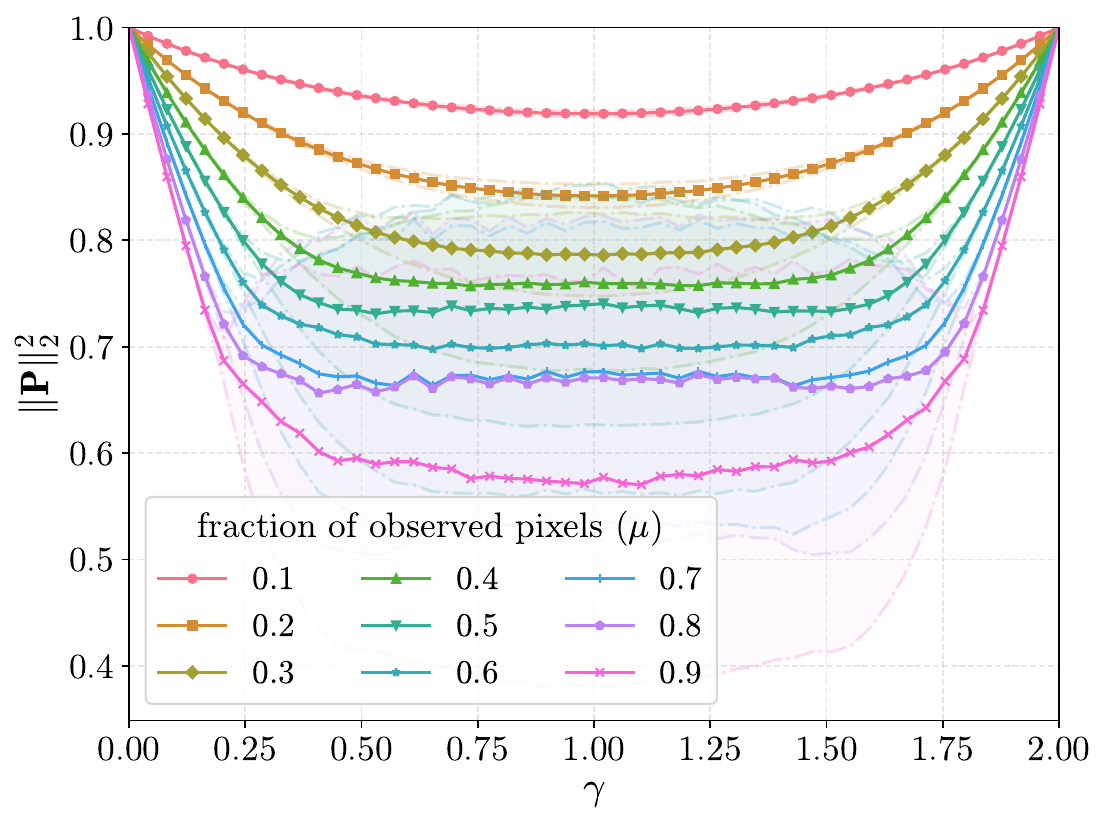}}\hspace{0.07\linewidth}
    \subfloat[\textbm{Sc-PnP-ISTA}]{\includegraphics[width=0.42\linewidth]{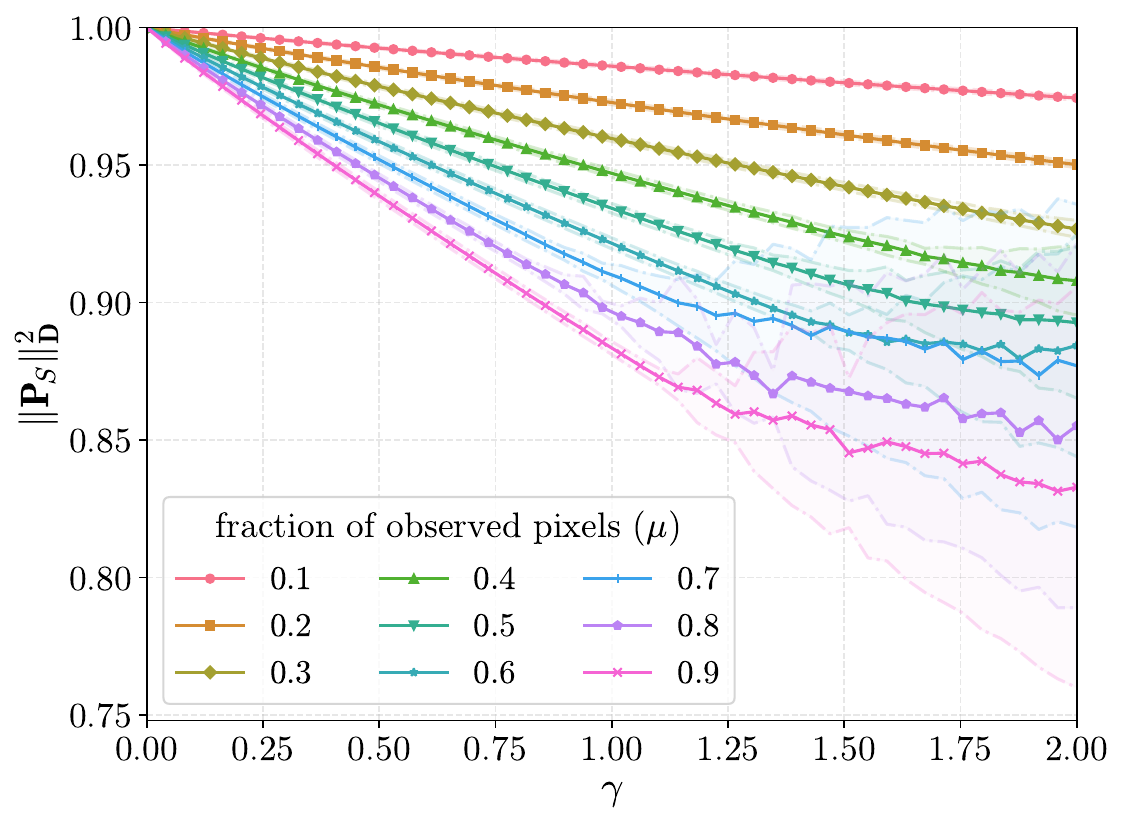}}
    \caption{We report the contraction factors of the update operators $\P$ and $\P_s$ in \textbm{PnP-ISTA} and \textbm{Sc-PnP-ISTA} for image inpainting. The contraction factors are shown for different fractions of observed pixels $\mu$ and for varying step-sizes $\gamma$. We note that the above operators are derived from the denoiser, which in turn is computed from the input image. The plots display the average results across images from Set12, with the shaded region around each plot representing the standard deviation. The results align with the predictions of Thm.~\ref{thm:boundISTA}, confirming that the contraction factor decreases with an increase in the fraction of observed pixels.}
    \label{fig:thm5-mu-effect}
\end{figure*}

\begin{figure*}
    \centering
    \subfloat[superresolution with \textbm{PnP-ADMM}]{
    \includegraphics[width=0.42\linewidth]{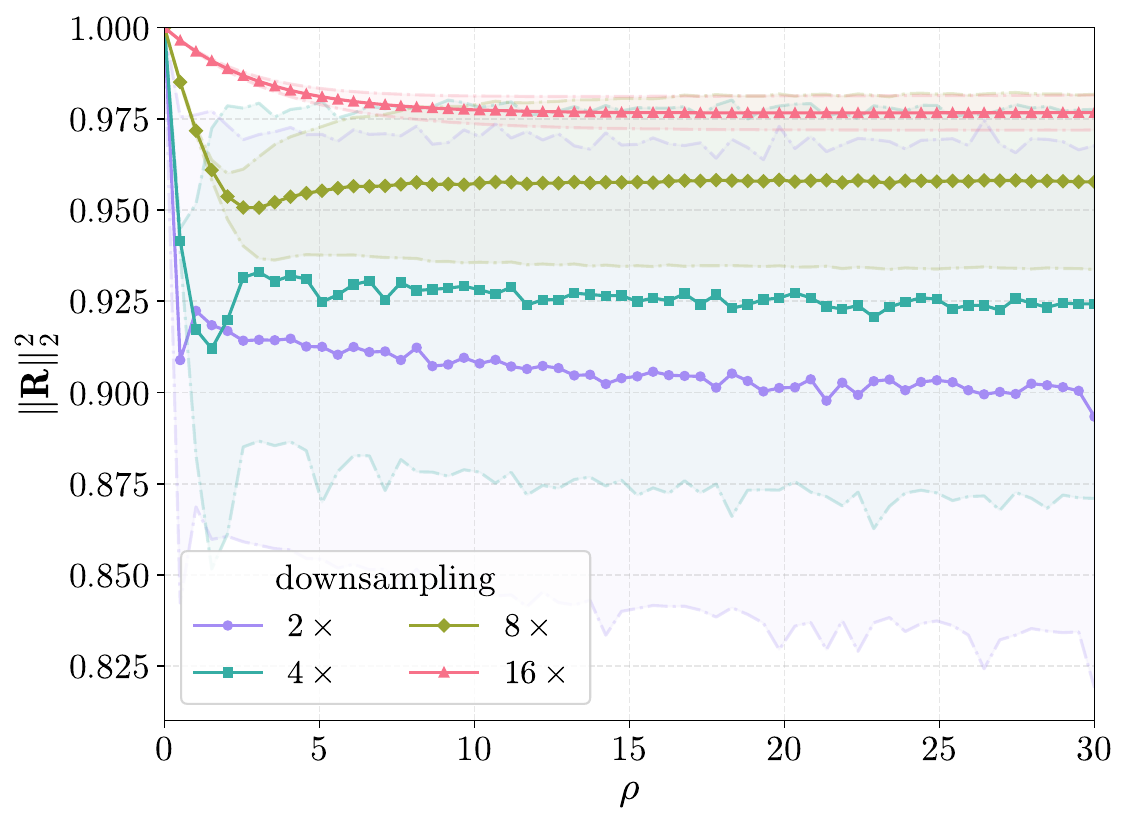}}
    \hspace{0.07\linewidth}
\subfloat[superresolution with \textbm{Sc-PnP-ADMM}]{\includegraphics[width=0.42\linewidth]{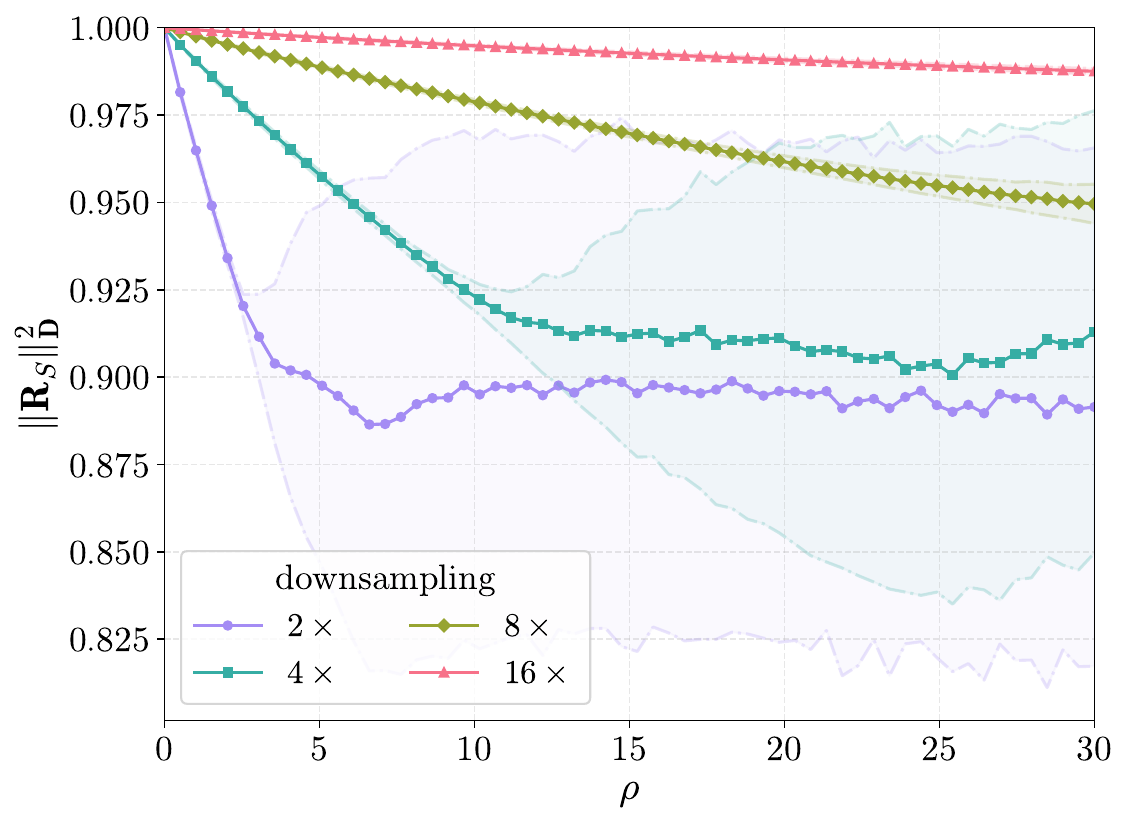}}
    \caption{We study the effect of the downsampling rate (for superresolution) on the contraction factor of the update operators: $\R$ for \textbm{PnP-ADMM} and $\R_s$ for \textbm{Sc-PnP-ADMM}. The bound in Thm. \ref{thm:inp-sr-Js} predicts a reduction in the contraction factor with lower downsampling (i.e., more samples). The results in the plots align well with this prediction. 
}
    \label{fig:thm8-mu-effect}
\end{figure*}

\begin{figure*}[h]
    \centering
    \subfloat[\text{inpainting}]{
    \includegraphics[width=0.42\linewidth]{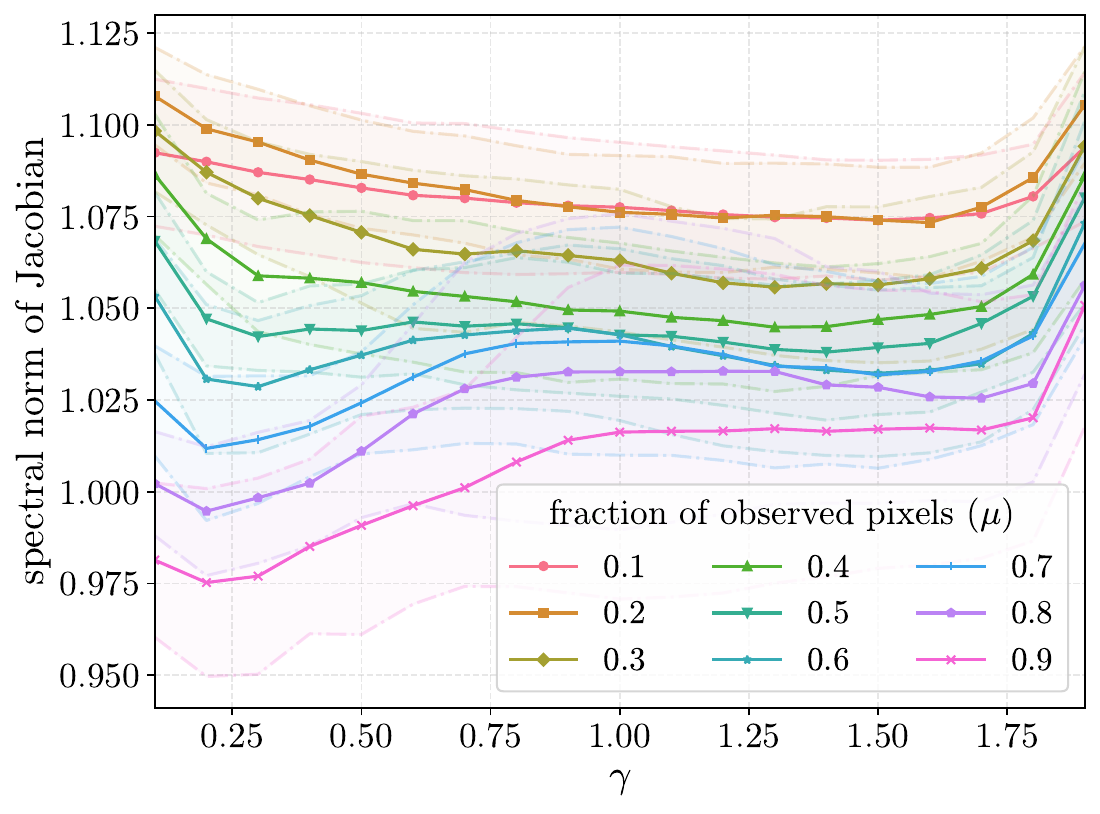}}\hspace{0.07\linewidth}
    \subfloat[\text{superresolution}]{\includegraphics[width=0.42\linewidth]{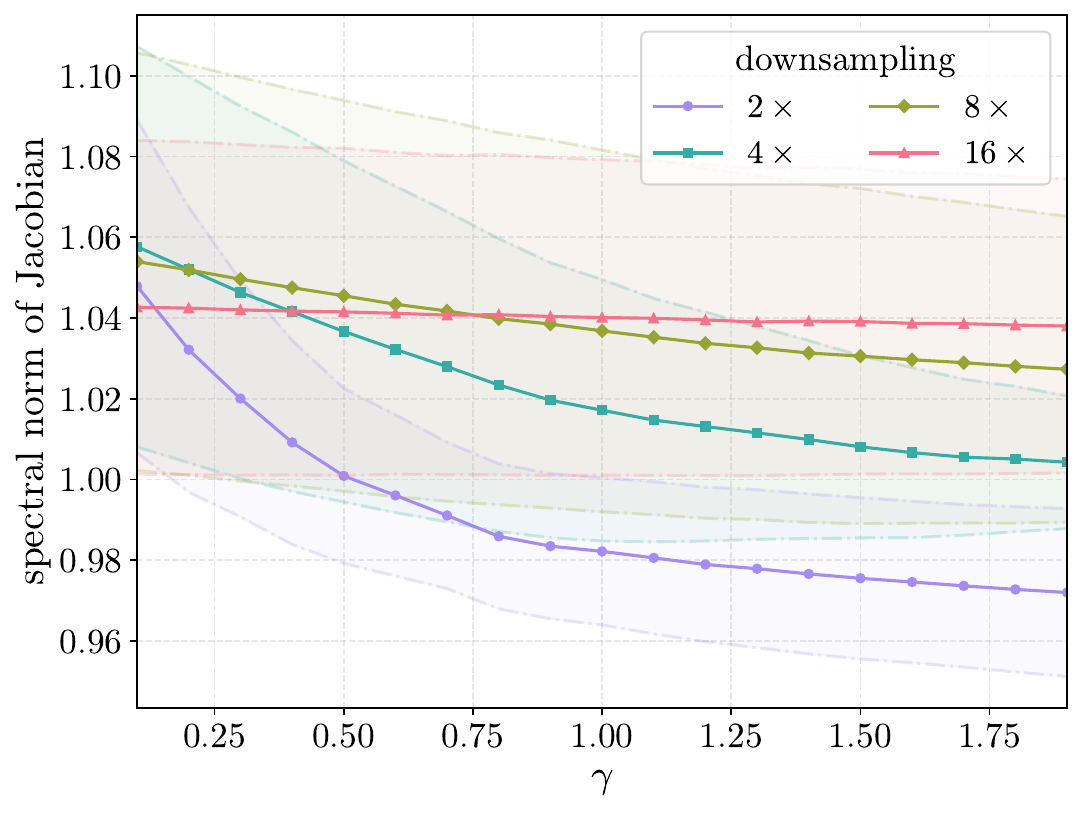}}
    \caption{We empirically validate the predictions of Theorem~\ref{thm:boundISTA} using a variant of the nonlinear DnCNN denoiser~\cite{pesquet_learning_2021}. Specifically, we compute the Jacobian of the update operator in \textbm{PnP-ISTA} at each iterate \(\x_k\), where k ranges from 0 to 100 in increments of 10, and take the maximum spectral norm across all such Jacobians. The plots show the average results across images from Set12. Interestingly, despite the nonlinearity of the denoiser, we observe a pattern similar to that in Fig.~\ref{fig:thm5-mu-effect}.}
    \label{fig:dncnn}
\end{figure*}

\begin{figure*}
\centering
\begin{minipage}[b]{.49\textwidth}
\centering
\includegraphics[width=0.80\linewidth]{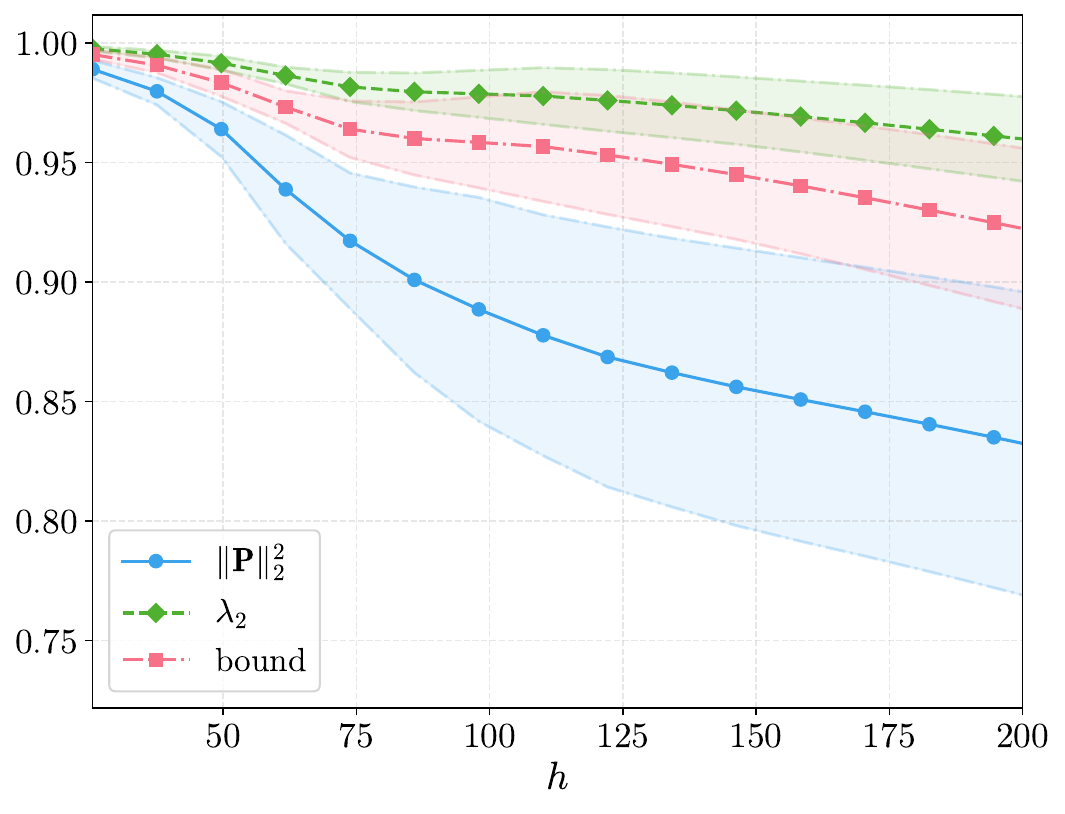}
    \caption{We study the correlation between the contraction factor of the update operator in \textbm{PnP-ISTA} and the second eigenvalue ($\lambda_2$) of the denoiser $\W$. The comparison is done for the Gaussian deblurring setup in Fig.~\ref{fig:deblurring}. We notice that as the bandwidth $h$ in \eqref{phi} increases, $\lambda_2$ decreases (i.e., the spectral gap widens). Notably, the experimental results shown here are consistent with Thm.~\ref{thm:boundISTA}, which predicts a reduction in (the bound on) the contraction factor (see~\eqref{bnd1}) with a decrease in $\lambda_2$, for a fixed step size $\gamma$.}
    \label{fig:thm5-lambda2-effect}
\end{minipage}\hfill
\begin{minipage}[b]{.49\textwidth}
\centering
\includegraphics[width=0.85\linewidth]{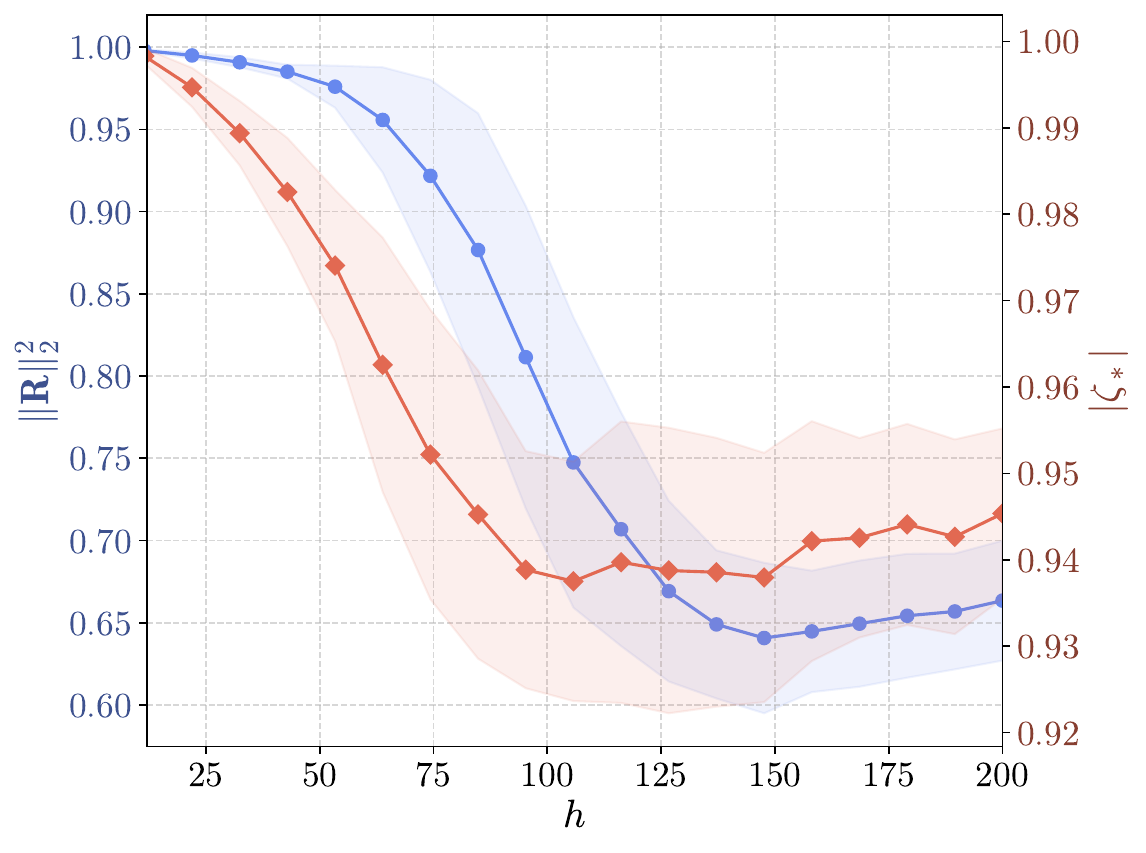}
    \caption{We study the effect of the spectral gap $1-|\zeta_*|$ on the contraction factor of the update operator for \textbm{PnP-ADMM}. The comparison is done for the Gaussian deblurring setup in Fig.~\ref{fig:deblurring}. We plot $\|\R\|_2$ and $\zeta_*$ on the left and right axes. As anticipated from the bounds in Thm. \ref{thm:inp-deblur-sym}, we observe a decrease in $\|\R\|_2$ as $|\zeta_*|$ decreases (i.e., larger spectral gap). As in Fig.~\ref{fig:thm5-lambda2-effect}, a reduction in $|\zeta_*|$, and hence the contraction factor, can generally be achieved by increasing the bandwidth $h$ of the denoiser.}
    \label{fig:zeta-effect}
\end{minipage}
\end{figure*}

\begin{figure*}
    \centering
    \subfloat[iterate convergence]{\includegraphics[width=0.42\linewidth]{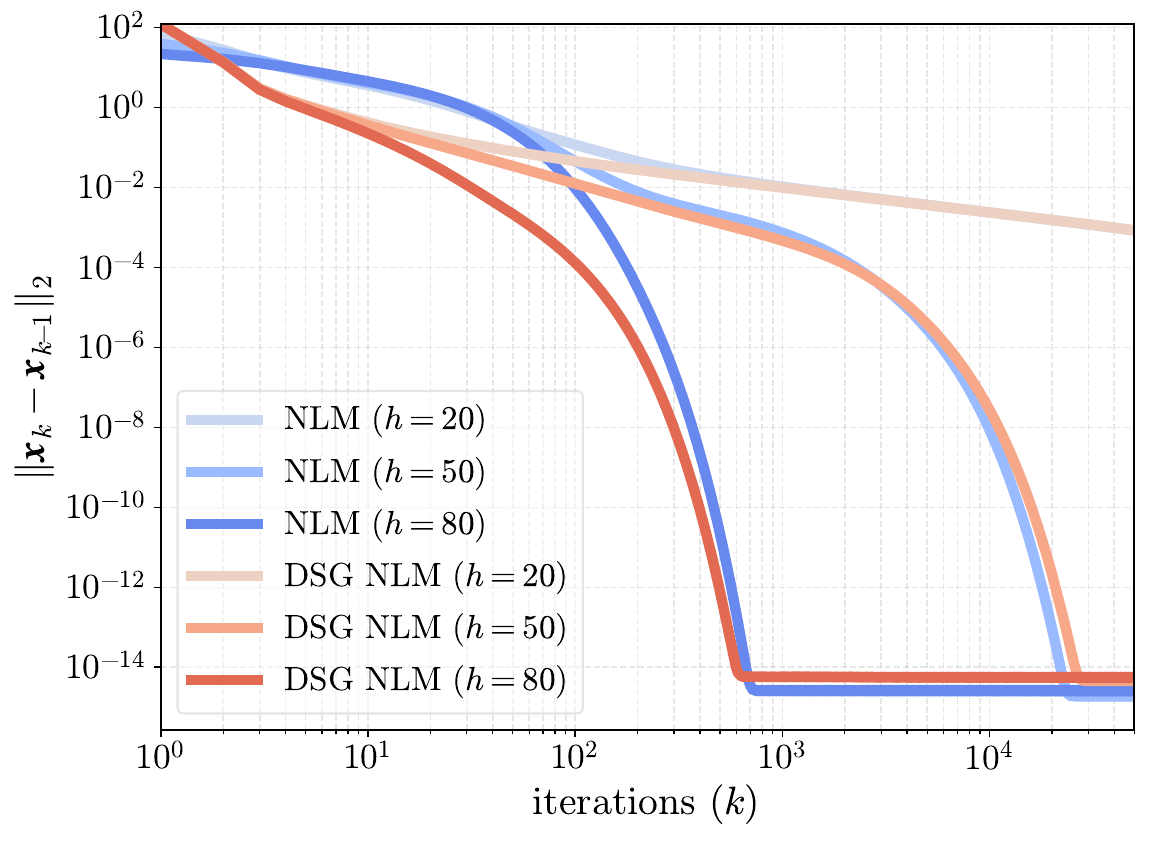}}\hspace{0.07\linewidth}
    \subfloat[PSNR convergence]{\includegraphics[width=0.42\linewidth]{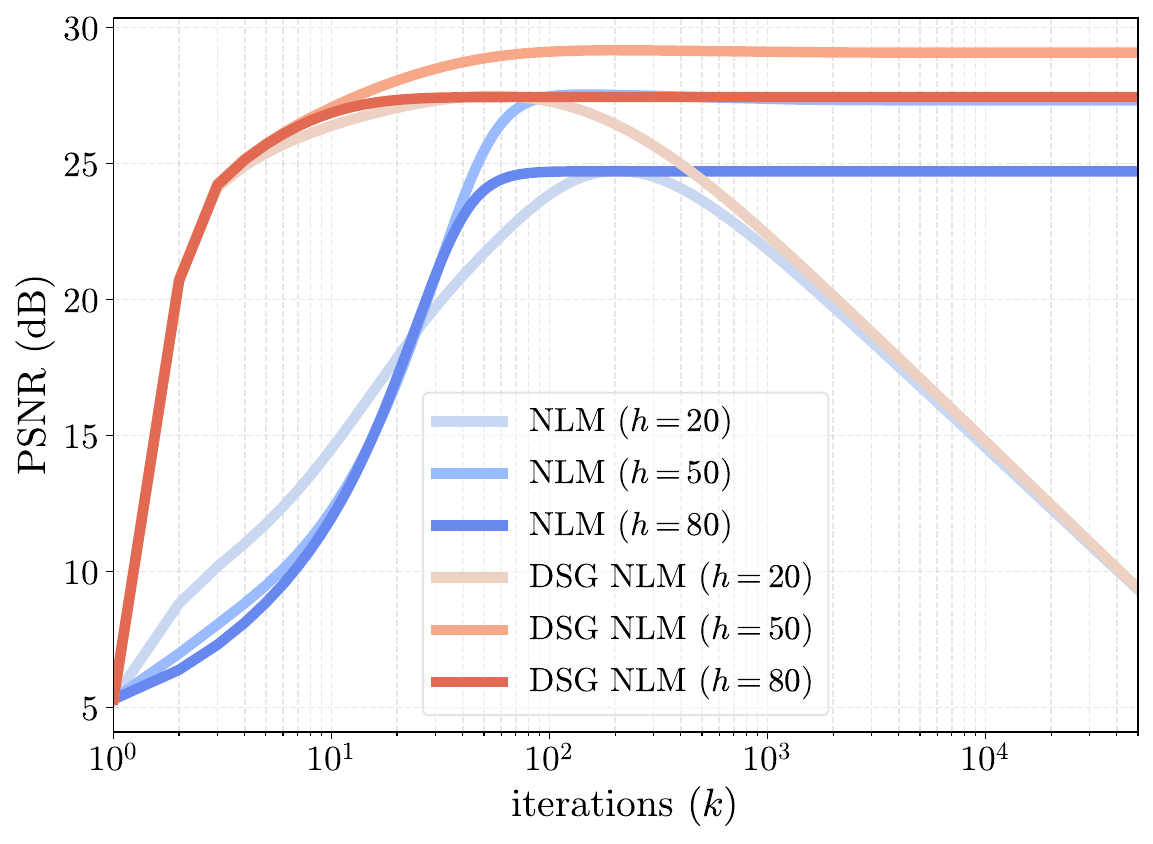}}
    \caption{We show the successive difference of the iterates generated by \textbm{PnP-ISTA} and \textbm{Sc-PnP-ISTA} with the symmetric DSG-NLM and the nonsymmetric denoiser NLM. This is for the Gaussian deblurring experiment in Fig.~\ref{fig:deblurring}. Also shown on the right is the evolution of the PSNR with iterations. We study the convergence behavior for different values of $h$, the bandwidth of the isotropic Gaussian kernel used in both denoisers. Since the final reconstruction is independent of the initialization (due to the contractive nature of the fixed-point operator), we have initialized using a noisy image. As discussed in Thm.~\ref{thm:boundISTA} and Fig. \ref{fig:thm5-lambda2-effect}, we observe that a higher bandwidth leads to a lower $ \lambda_2 $ and faster convergence. The convergence is slow for a small bandwidth of $h=20$ but improves significantly as $h$ increases. However, there is a limit to this improvement, as increasing the bandwidth indefinitely can degrade reconstruction quality. As seen in Fig.~(b), the best reconstruction occurs at $h=50$, with further increases in $h$ accelerating convergence but decreasing reconstruction quality.}
    \label{fig:convergence-plot}
\end{figure*}

\section{Numerical Experiments}
\label{sec:exp}

We present numerical results\footnote{The codes are available at \texttt{github.com/arghyasinha/tsp-kernel-denoiser}.} to validate our key results, specifically Thms.~\ref{thm:boundISTA}, \ref{thm:inp-deblur-sym}, and \ref{thm:inp-sr-Js}. Moreover, we wish to investigate the dependence of the contraction factor––which controls the convergence rate––on the following factors: the second eigenvalue $\lambda_2$ of $\W$ for \textbm{PnP-ISTA}, the second dominant eigenvalue $\zeta_*$ of $\V$ for \textbm{PnP-ADMM}, the fraction of observed pixels $\mu$ for inpainting, and the downsampling rate for superresolution. It is important to note that these experiments are not intended to assess the reconstruction accuracy of \textbm{PnP-ISTA} and \textbm{PnP-ADMM} with kernel denoisers, nor to compare them with trained denoisers; for such comparisons, we refer to~\cite{nair2022plug,nair2019hyperspectral,sreehari2016plug,Teodoro2019PnPfusion}. Instead, we aim to provide numerical evidence supporting the theoretical findings.

Before proceeding with the result, we discuss the computation of  $\|\P\|_2=\sigma_{\max}(\P)$, the largest singular value of the update operator in \textbm{ISTA}. Since \(\W\) and \(\A\) are too large to be manipulated as matrices, we cannot store $\P$ as a matrix and use a direct method for finding $\sigma_{\max}(\P)$. Instead, we treat them as operators (or black boxes) and use the power method \cite{Horn_Johnson_1985}. Similarly, we use the power method to compute \(\|\R\|_2\), \(\|\P_s\|_\D\), and \(\|\R_s\|_\D\). Before applying the power method to the scaled variants \(\|\P_s\|_\D\) and \(\|\R_s\|_\D\), we performed an additional step of changing the norm from \(\|\cdot\|_\D\) to \(\|\cdot\|_2\) using \eqref{D-ip-norm}. To compute the second eigenvalue \(\lambda_2\), we applied the power method to the deflated operator \(\W - (1/n)\e\e^\top\). This approach is effective because \(1\) is the leading eigenvalue of \(\W\), with a one-dimensional eigenspace spanned by \(\e\) (see Prop. \ref{propertiesW}). Consequently, the deflated operator has \(\lambda_2\) as its dominant eigenvalue. Similarly, we computed \(\zeta_*\) using the deflated operator \(\V - (1/n)\e\e^\top\).

We perform inpainting, deblurring, and superresolution experiments. Unless stated otherwise, the forward operator $\A$ corresponds to randomly sampling $30\%$ pixels for inpainting, and a $25 \times 25$ Gaussian blur with standard deviation $1.6$ for deblurring. For superresolution,  $\A = \S\B$, where we take \(\B\) is a \(9 \times 9\) uniform blur and \(\S\) coresponds to \(2\times\) downsampling. As for the denoiser, we consider the nonsymmetric NLM and the symmetric DSG-NLM denoisers discussed in Sec. \ref{sec:kd}. Specifically, the key parameter in our experiments is the bandwidth $h$ (refer to Sec. \ref{sec:kd}), the standard deviation of the underlying isotropic Gaussian kernel used in these denoisers.

We present deblurring and superresolution results using \textbm{PnP-ISTA} and \textbm{Sc-PnP-ISTA} in Figs.~\ref{fig:deblurring} and \ref{fig:superresolution}. Additionally, in Fig.~\ref{fig:bar_plot_symm_vs_kernel}, we compare reconstruction across various images from the Set12 dataset. The results suggest that DSG-NLM generally outperforms its nonsymmetric counterpart. In Fig.~\ref{fig:thm5-mu-effect}, we notice how the contraction factor improves as we increase the fraction of observed pixels in inpainting. A similar trend is observed for superresolution in Fig.~\ref{fig:thm8-mu-effect}. These findings align with the behavior predicted by the bounds in Thms.~\ref{thm:boundISTA} and \ref{thm:inp-sr-Js}.

In Fig.~\ref{fig:dncnn}, we expand our analysis to include a nonlinear trained denoiser~\cite{pesquet_learning_2021}. 
However, since the denoiser is nonlinear, we need to work with the Jacobian of the update operator in~\eqref{eq:pnp-ista}––with $\W$ as the nonlinear denoiser––and its spectral norm. Despite the nonlinearity, we empirically observe in Figs.~\ref{fig:thm5-mu-effect} and \ref{fig:thm8-mu-effect} that the spectral norms follow a similar pattern. However, since the spectral radius occasionally exceeds 1, the operator cannot be expected to be locally contractive.

Finally, we study the impact of parameters derived from the denoiser on the convergence. In Fig. \ref{fig:thm5-lambda2-effect}, we show that a denoiser with a lower \(\lambda_2\) results in a smaller contraction factor, consistent with the predictions of Thm.~\ref{thm:boundISTA}. Similarly, Fig.~\ref{fig:zeta-effect} highlights the relation between the contraction factor of \textbm{PnP-ADMM} and \(\zeta_*\), as anticipated in Thm. \ref{thm:inp-deblur-sym}.
We decrease $\lambda_2$ (and $\zeta_*$) by increasing the bandwidth $h$ of the Gaussian kernel.
Note that Thms. \ref{thm:boundISTA} and \ref{thm:inp-sr-Js} do not allow for similar predictions regarding \textbm{Sc-PnP-ISTA} and \textbm{Sc-PnP-ADMM} as the bounds are influenced not only by \(\lambda_2\) (or \(\zeta_*\)) but also by \(\D\). Since both \(\lambda_2\) (or \(\zeta_*\)) and $\D$ are derived from $\W$, in the process of changing \(\lambda_2\) we also end up altering $\D$––it is not possible to modify one while keeping the other fixed. This contrasts with the symmetric case, which does not exhibit this complex coupling.
However, for smaller values of $h$, where changes in $\D$ have a minimal impact, we observe a reduction in the contraction factor as $\lambda_2$ decreases (see Fig. \ref{fig:convergence-plot}). In Fig. \ref{fig:convergence-plot}, we demonstrate the effects of varying the bandwidth $h$ on the convergence of both \textbm{PnP-ISTA} and \textbm{Sc-PnP-ISTA} in practice. Specifically, we observe that while increasing $h$ can significantly speed up convergence for both algorithms, an excessive increase in $h$ may lead to subpar reconstructions.

\section{Conclusion}

We proved that the iterates of PnP-ISTA and PnP-ADMM converge globally and linearly to a unique reconstruction when using generic kernel denoisers. This does not follow from standard convex optimization theory since the associated loss function is not strongly convex. The key insight is that the contractivity is not with respect to the standard Euclidean norm but rather a norm induced by the kernel denoiser. This requires modifications to the definition of the gradient in PnP-ISTA and the proximal operator in PnP-ADMM. We also performed experiments with the archetype kernel denoiser, NLM, and its symmetric counterpart, DSG-NLM. We found that the reconstruction quality is typically better for DSG-NLM than NLM. However, DSG-NLM requires about $3\times$ more computations than NLM. Thus, there is a tradeoff between per-iteration cost and reconstruction quality within the class of kernel denoisers.

We derived quantitative bounds on the convergence rate for inpainting, deblurring, and superresolution. The primary challenge was identifying a bound that is $<1$. To achieve this, we established various relationships between the parameters of the denoiser and the forward operator. For ADMM, an additional challenge involved evaluating an inverse operator, which we approximated by exploiting specific properties of the forward operator. Nonetheless, there is scope for improving the bounds and using them to optimize the algorithmic parameters.

The derived bounds provide critical insights into the interactions among the different algorithmic components. Notably, we found that the spectral gap of the denoiser significantly influences the convergence rate: the larger the gap, the faster the convergence. However, artificially increasing the spectral gap by adjusting denoiser parameters (such as the bandwidth) is contingent upon the forward operator and may not work in all cases. An interesting question is whether we can design a kernel denoiser with an enhanced spectral gap and deploy it within PnP algorithms to achieve faster convergence without compromising the reconstruction quality.
\section*{Acknowledgement} We sincerely thank the editor and the reviewers for their suggestions that helped to improve the paper a lot.

\section*{Appendix}

\subsection{Proof of Lemma~\ref{MN-lemma-1}}
\label{pf:MN-lemma-1}

Let $\M$ and $\N$ be as in Lemma~\ref{MN-lemma-1}. Note that $\|\M\N\| <1 $ is equivalent to requiring that $\|\M\N  \x\| < \|\x\|$ for all nonzero $\x \in \Re^n$. Now, since $\sigma(\M) \subset (-1,1]$, we have $\|\M \z \| \leqslant \|\z\|$ for all $\z \in \Re^n$. Therefore, if $\x \notin \fix(\N)$, then $\|\N\x\| < \|\x\|$ by Prop.~\ref{prop:ctr-fix}, and hence
\begin{equation*}
\|\M \, (\N  \x) \| \leqslant \| \N  \x \| < \|\x\|.
\end{equation*}
On the other hand, for any $\x \in \fix(\N)$, we have $\|\M (\N   \x)\| = \|\M\x\|$. However, from  $\fix(\M) \cap \fix(\N) =\{\ze\}$, we can conclude that $\x \notin \fix(\M)$, and hence, by Prop.~\ref{prop:ctr-fix}, $\|\M\x\| < \|\x\|$. This completes the proof of Lemma~\ref{MN-lemma-1}. 

\subsection{Proof of Lemma~\ref{MN-lemma-2}}
\label{pf:MN-lemma-2}

Since $\M$ is self-adjoint, we have
\begin{equation*}
 \norm{\M\z}^2 = \sum_{i=1}^n \lambda_i^2 \inner{\z}{\q_i}^2,
\end{equation*}
for all $\z \in \Re^n$, where $\q_1,\ldots,\q_n$ is an orthonormal basis of eigenvectors of $\M$ with eigenvalues $\lambda_1,\ldots,\lambda_n$. In particular, for any $\x \in \Re^n$, 
\begin{equation*}
\norm{\M\N\x}^2=   \norm{\M(\N\x)}^2 = \sum_{i=1}^n \lambda_i^2 \inner{\N\x}{\q_i}^2.
\end{equation*}
Now, using \eqref{order-lambda}, we can write
\begin{equation}
\label{eq:MNx}
    \norm{\M\N\x}^2  \leqslant  \lambda_1^2\inner{\N\x}{\q_1}^2 + \lambda_2^2\sum_{i=2}^n \inner{\N\x}{\q_i}^2.
\end{equation}
Furthermore, we can regroup the right side of \eqref{eq:MNx} into
\begin{align*}
(\lambda_1^2 - \lambda_2^2)\inner{\N\x}{\q_1}^2& +\lambda_2^2\sum_{i=1}^n \inner{\N\x}{\q_i}^2 \\
    &= (\lambda_1^2 - \lambda_2^2)\inner{\N\x}{\q_1}^2 +\lambda_2^2 \, \norm{\N\x}^2. 
\end{align*}
Now, $\inner{\N\x}{\q_1}=\inner{\x}{\N\q_1}$, and $|\inner{\x}{\N\q_1}| \leqslant  \norm{\x}\, \norm{\N \q_1}$ by Cauchy-Schwarz. Finally, from $\sigma(\N) \subset (-1,1]$,  we have $\norm{\N\x} \leqslant \norm{\x}$. Thus, we obtain 
\begin{equation*}
    \norm{\M\N\x}^2 \leqslant \big(  (\lambda_1^2 - \lambda_2^2)\norm{\N \q_1}^2 +\lambda_2^2 \big) \norm{\x}^2.
\end{equation*}
Since this holds for any $\x \in \Re^n$, we have \eqref{eq:MNbound}.

\subsection{Proof of Thm.~\ref{thm:inp-bound-Js}}
\label{pf:inp-bound-Js}

We use Lemma~\ref{MN-lemma-2} with $\M=\V$ and $\N=\F_s$. We can check that $\V$ and $\F_s$ are self-adjoint on $(\Re^n,\langle \cdot, \cdot \rangle_\D)$. Moreover, $|\lambda_1|=1, \q_1=\e/\|\e\|_\D$ , and $|\lambda_2|= \zeta_*$ (see~\eqref{def:zeta}).   

To bound $\|\F_s  \q_1\|^2_\D$, we need to compute $\|\F_s  \e\|^2_\D$. Let $\Omega$ be as in the proof Thm.~\ref{thm:inp-deblur-sym}. Since $\F_s$ is diagonal, we have
\begin{equation*}
(\F_s\e)_i =
 \begin{cases}
     (1-\rho \D_{ii}^{-1})/(1+\rho \D_{ii}^{-1}) & \quad  i \in \Omega, \\
     1 & \quad   i \notin \Omega.
\end{cases}
\end{equation*}
In particular, let $\theta = \max_{i}\, (1-\rho \D_{ii}^{-1})/(1+\rho \D_{ii}^{-1}) <1$. Then we have
\begin{align*}
\|\F_s \e\|^2_\D &= \sum_{i \notin \Omega} \D_{ii}\cdot 1 + \theta^2 \sum_{i \in \Omega} \D_{ii} \\
& \leqslant  \mathrm{trace}(\D) - (1-\theta^2) \sum_{i \in \Omega} \D_{ii}.
\end{align*}
Thus, from $\|\e\|_\D^2=\mathrm{trace}(\D)$ and $1 \leqslant \D_{ii} \leqslant \|\D\|_2$, we have
\begin{align*}
\|\F_s  \q_1\|^2_\D & \leqslant 1 -  (1-\theta^2) \frac{\sum_{i \in \Omega} \D_{ii}}{\mathrm{trace}(\D)}\\
 &\leqslant 1 -  (1-\theta^2) \frac{\mu}{\|\D\|_2}.
\end{align*}
Substituting this in Lemma~\ref{MN-lemma-2}, we get \eqref{eq:inp-bound-Js}.

\bibliographystyle{IEEEtran}
\bibliography{refs}

\end{document}